\documentclass[orivec]{llncs}

\usepackage{header}
\usepackage{abbreviations}
\usepackage{symbols}
\let\union\relax
\usepackage{defs}
\usepackage{versioning}

\begin{document} 

\title{Abstract Interpretation with Unfoldings\thanks{Supported
by ERC project 280053 (CPROVER) and a Google Fellowship.}}

\author{
Marcelo Sousa\inst{1},
C{\'e}sar Rodr{\'\i}guez\inst{2},
Vijay D'Silva\inst{3}
\and
Daniel Kroening\inst{1}}
\authorrunning{Marcelo Sousa \and C{\'e}sar Rodr{\'\i}guez \and Vijay D'Silva \and Daniel Kroening }

\institute{
University of Oxford, United Kingdom
\and
Universit\'e Paris 13, Sorbonne Paris Cit\'e, LIPN, CNRS, France
\and
Google Inc., San Francisco
}

\maketitle  
\pagenumbering{arabic}

\begin{abstract}
We present and evaluate a technique for computing path-sensitive
interference conditions during abstract interpretation of concurrent
programs.
In lieu of fixed point computation, we use prime event structures
to compactly represent causal dependence and interference between
sequences of transformers.
Our main contribution is an unfolding algorithm that uses
a new notion of independence to avoid redundant transformer
application,
thread-local fixed points to reduce the size of the unfolding,
and a novel cutoff criterion based on subsumption to guarantee 
termination of the analysis.
Our experiments show that the abstract unfolding produces
an order of magnitude fewer false alarms than a mature abstract interpreter,
while being several orders of magnitude faster than solver-based tools 
that have the same precision.
\end{abstract}

\section{Introduction}
\label{sec:intro}
This paper is concerned with the problem of extending an abstract
interpreter for sequential programs to analyze concurrent programs.
A na\"ive solution to this problem is a global fixed point analysis
involving all threads in the program. 
A~distinct solution is to analyze threads in isolation and exchange
invariants on global variables between
threads~\cite{Mine14,Mine12,CH09}.
Related research, including this paper, 
seeks analyses that preserve the scalability of the local
fixed point approach without losing the precision of a
global fixed point.

We design and implement an \emph{abstract unfoldings}
data structure and analysis that combines an abstract domain 
with the type of unfolding algorithm used to analyze Petri nets.
An unfolding is a tree-like structure that uses partial orders to
represent concurrent executions and conflict relations to represent
interference.
A challenge in combining unfoldings with abstract domains is that
abstract domains typically provide approximations of states and
transitions, not traces, and are not equipped with interference
information.
Another challenge is that unfolding algorithms are typically applied
to explicit-state analysis of systems with deterministic transitions
while abstract domains are symbolic and non-deterministic
owing to abstraction.

The main idea of this paper is to construct an unfolding \emph{of an 
analyzer}, rather than a program.
An event is the application of a transformer in an analysis context, and
concurrent executions are replaced by a partial order on transformer
applications.
We introduce independence for transformers and use this notion to
construct an unfolding of a domain given a program and independence
relation.
The unfolding of a domain is typically large and we use thread-local
fixed point computation to reduce its size without losing interference
information.

From a static analysis perspective, our analyser is a path-sensitive
abstract interpreter that uses an independence relation to compute a
history abstraction (or trace partition) and 
organizes exploration information in an unfolding.
%
From a dynamic analysis perspective, our approach is a super-optimal
POR~\cite{RSSK15} that
uses an abstract domain to collapse branches of the computation tree originating
from thread-local control decisions.
\subsubsection*{Contribution}
We make the following contributions towards reusing an abstract interpreter
for sequential code for the analysis of a concurrent program.
\begin{enumerate}
\item
  A new notion of transformer independence for unfolding with domains
  (\secref{ind}).
\item
  The unfolding of a domain, which provides a sound way to combine
 transformer application and partial-order reduction
  (\secref{unf.def}).
\item
  A method to construct the unfolding using thread-local
  analysis and pruning techniques 
  (\secref{cutoffs}, \secref{collapse}).
\item
  An implementation and empirical evaluation demonstrating the
  trade-offs compared to an abstract interpreter and solver-based tools
  (\secref{exp}).
\end{enumerate}

We provide the proofs of our formal results in the Appendix.

\section{Motivating Example and Overview}
\label{sec:overview}

\begin{figure}[t]
\centering
\includegraphics[scale=0.85]{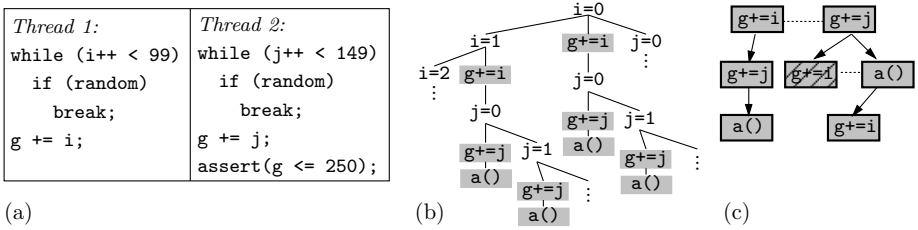}
\caption{(a) Example program (b)~Its POR exploration tree (c)~Our
unfolding}
\label{fig:forfor}
\end{figure}

Consider the program given in \figref{forfor}~(a), which we wish to prove
safe using an interval analysis.
Thread~1 (\resp 2) increments~\verb!i! (\resp \verb!j!) in a loop that can
non-deterministically stop at any iteration.
All variables are intialized to~0 and the program is safe as the
\verb!assert! in thread~2 cannot be violated.

When we use a POR approach to prove safety of this program, the exploration 
algorithm exploits the fact that only the interference between statements
that modify the variable \verb!g! can lead to distinct final states.
This interference is typically known as~\emph{independence}~\cite{Pel93,God96}. 
The practical relevance of independence is that one can use it
to define a safe fragment, given in~\figref{forfor}~(b), of 
the computation tree of the program
which can be efficiently explored~\cite{RSSK15,AAJS14}.
At every iteration of each loop, the conditionals open one more branch in the
tree. 
Thus, each branch contains a different write to the global variable, which
is dependent with the writes of the other thread as the order of their application
reaches different states. 
As a result, the exploration tree becomes intractable very fast.
It is of course possible to bound the depth of the exploration at the expense 
of completeness of the analysis.

The thread-modular static analysis that is implemented in
\astreea{}~\cite{Mine14} or \mbox{\framac{}}~\cite{YaB12} incorrectly
triggers an alarm for this program.
These tools statically analyze each thread in isolation assuming that
\verb!g! equals~0.
Both discover that thread~1 (\resp~2) can write $[0,100]$
(\resp~$[0,150]$) to \verb!g! when it reads~0 from~it.
Since each thread can modify the variable read by the other, they repeat
the analysis starting from the join of the new interval with the initial
interval.
In this iteration, they discover that thread~2 can write $[0,250]$ to \verb!g! when it
reads $[0,150]$ from it.
The analysis now incorrectly determines that it needs to re-analyze
thread~2, because thread~1 also wrote $[0,250]$ in the previous iteration and
that is a larger interval than that read by thread~2.
This is the reasoning behind the false alarm.
The core problem here is that these methods are path-insensitive across
thread context switches and that is insufficient to prove this assertion.
The analysis is accounting for a thread context switch
that can never happen (the one that flows~$[0,250]$ to thread~2 before thread~2
increments~\verb!g!).
More recent approaches~\cite{KW16,MM17} can achieve a higher degree of
flow-sensitivity but they either require manual annotations to guide the
trace partitioning or are restricted to program locations outside of a loop
body.

Our key contribution is an unfolding that is flow- and path-sensitive across
interfering statements of the threads and
path-insensitive inside the non-interfering blocks of statements.
\Cref{fig:forfor}~(c) shows the unfolding structure 
that our method explores for this program.
The boxes in this structure are called \emph{events} and 
they represent the action of firing a transformer after a history of firings.
The arrows depict~\emph{causality} constraints between events, \ie, the
\emph{happens-before} relation.
Dotted lines depict the immediate \emph{conflict relation}, stating that two
events cannot be simultaneously present in the same concurrent execution, known 
as \emph{configuration}.
This structure contains three maximal configurations (executions),
which correspond to the three ways in which the statements reading 
or writing to variable~\verb!g! can interleave.

Conceptually, we can construct this unfolding using the following idea:
start by picking an arbitrary interleaving.
Initially we pick the empty one which reaches the initial state of
the program.
Now we run a sequential abstract interpreter on one thread, say thread~1, 
from that state and stop on every location that reads or writes a global 
variable.
In this case, the analyzer would stop at the statement \verb!g += i! 
with the invariant that~$\tup{g \mapsto [0,0], i \mapsto [0,100]}$.
This invariant corresponds to the first event of the unfolding (top-left corner).
The unfolding contains now a new execution, so we iterate again the same
procedure by picking the execution consisting of the event we just discovered.
We run the analyser on thread~2 from the invariant reached by
that execution and stop on any global action. 
That gives rise to the event~\verb!g+=j!, and in the next step 
using the execution composed of the two events we have seen, we 
discover its causal successor \verb!a()!.
Note however that before visiting that event, we could have added 
event~\verb!g+=j! corresponding to the invariant of running 
an analyser starting from the initial state on thread~2.
Furthermore we know that because both invariants are related to
the same shared variable, these two events must be ordered.
We enforce that order with the conflict relation.

Our method mitigates the aforementioned branching 
explosion of the POR tree because it never unfolds 
the conflicting branches of a naive exploration.
In~comparison to thread-modular analysis,
it remains precise about the context switches 
because it uses a history-preserving data structure.

Another novelty of our approach is the observation that 
certain events are~\emph{equivalent} in the sense that 
the state associated with one is~\emph{subsumed} by the second.
In our example, one of these events, known as a~\emph{cutoff event},
is labelled by \verb!g+=i! and denoted with a striped pattern.
Specifically, the configuration $\set{\texttt{g+=i}, \texttt{g+=j}}$ reaches the
same state as $\set{\texttt{g+=j}, \texttt{g+=i}}$.
Thus, no causal successor of a cutoff event needs to be explored 
as any action that we can discover from the cutoff event can be found
somewhere else in the structure.

\paragraph{Outline.}
The following diagram displays the various concepts and transformations 
presented in the paper:

\begin{figure}
\input{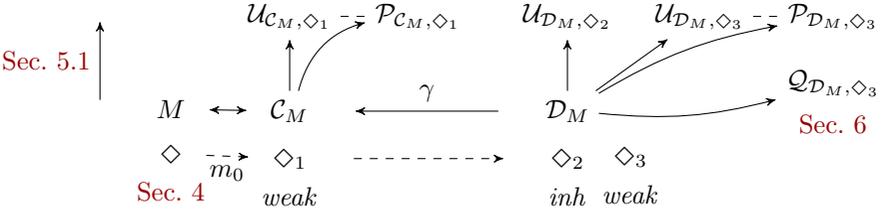}
\begin{center}
\begin{tikzpicture}[->,>=stealth',
 shorten >=0.5pt,
 node distance=30pt,
 auto,
 font=\normalsize
 ]

\begin{scope}[anchor=north]

\matrix[row sep=20pt, column sep=10pt]
{
  \node (s0) {};
  &
  \node (s1) [label=left:{}] {$\unf{\ccdom M,\indep_1}$};
  &
  \node (s11) {$\pref{\ccdom M,\indep_1}$};
  &
  \node (s111) [label=left:{}]  {$\unf{\ddom_M,\indep_2}$};
  &
  \node (v1) {$\unf{\ddom_M,\indep_3}$};
  &
  \node (v2) {$\pref{\ddom_M,\indep_3}$};
  \\
  \node (s)  {$M$};
  &
  \node (s2)  {$\ccdom M$};
  & 
  &
  \node (s3) {$\ddom_M$};
  \\
};
\node (v3) [below = 0.3cm of v2, label=below:{\cref{sec:collapse}}] {$\preff{\ddom_M,\indep_3}$};
\node (i1) [below = 0.1cm of s, label=below:{\cref{sec:ind}}] {$\indep$};
\node (i2) [below = 0.1cm of s2,label=below:{\emph{weak}}] {$\indep_1$};
\node (i3) [below = 0.1cm of s3,label=below:{\emph{inh}}] {$\indep_2$};
\node (i4) [below right = 0.1cm of s3,label=below:{\emph{weak}}] {$\indep_3$};
\node (sl) [left = 0.5cm of s] {};
\node (su) [left = of s11, above = of sl] {};

\draw[->] (sl) -- node [left] {\cref{sec:unf.def}} (su);

\path[poset-order]
  (s1) edge (s11)
  (v1) edge (v2)
  ;
\path[->] 
  (s2) edge [bend left] (s11)
	   edge (s1)
  (s3) edge [shorten <=0.5cm, shorten >=0.5cm] node [above] {$\gamma$} (s2)
  (i1) edge [shorten <=0.2cm, shorten >=0.2cm, dashed] node {$m_0$} (i2)
  (i2) edge [shorten <=0.5cm, shorten >=0.5cm, dashed] (i3)
  (s3) edge (s111)
	   edge (v1)
	   edge [bend left=10] (v2)
	   edge [bend right=10] (v3)
;
\path[<->]
  (s) edge [shorten <=0.2cm, shorten >=0.2cm] (s2)
;
\end{scope}
\end{tikzpicture}
\end{center}
\caption{Overview diagram}
\label{fig:overview}
\end{figure}

Let~$M$ be the program under analysis whose concrete 
semantics~$\cdom_M$ is abstracted by a domain~$\ddom_M$.
The relations~$\indep$ and~$\indep_i$ are independence relations
with different levels of granularity 
over the transformers of~$M$, $\ccdom M$, or $\ddom_M$.
We denote by $\unf{\ddom',\indep'}$ the \emph{unfolding} of 
either~$\ccdom M$ or $\ddom_M$ under independence relation~$\indep'$
(defined in \cref{sec:unf.def}).
Whenever we unfold a domain using a weak independence relation
($\indep_2$ on~$\ccdom M$ and $\indep_3$ on~$\ddom_M$),
we can use cutoffs to prune the unfolding represented
by the dashed line between unfoldings.
The resulting unfolding (defined in~\cref{sec:cutoffs})
is denoted by the letter~$\ppref$.
The main contribution of our work is the 
\emph{compact unfolding}, $\preff{\ddom_M,\indep_3}$, 
described above.

\section{Preliminaries}
\label{sec:prelim}

\begin{figure}[h]
\centering
\includegraphics[width=\textwidth]{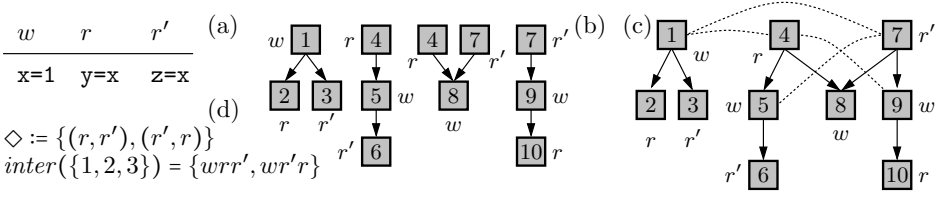}
\caption{Unfolding of a simple program}
\label{fig:exunf}
\end{figure}

There is no new material in this section,
but we recommend the reader to review the definition of an analysis
instance, which is not standard.

\paragraph{Concurrent programs.}
We model the semantics of a concurrent, non-deterministic program by a
labelled transition system
$M \eqdef \tup{\Sigma, \to, A, s_0}$, where
$\Sigma$ is the set of \emph{states},
$A$ is the set of \emph{program statements},
${\to} \subseteq \Sigma \times A \times \Sigma$ is the transition relation, and
$s_0$ is the \emph{initial state}.
The identifier of the thread containing a statement~$a$
is given by a function $p \colon A \to \nat$.
If $s \fire{a} s'$ is a transition, 
the statement $a$ is \emph{enabled} at $s$, 
and $a$ can \emph{fire} at $s$ to produce $s'$.
We let~$\enabl s$ denote the set of statements 
enabled at~$s$.
As statements may be non-deterministic, 
firing $a$ may produce more than one such $s'$.
%
%
A sequence $\sigma \eqdef a_1 \ldots a_n \in A^*$ is a \emph{run} when
there are states $s_1, \ldots, s_n$ satisfying
$s_0 \fire{a_1} s_1 \ldots \fire{a_n} s_n$.
For such $\sigma$ we define $\state \sigma \eqdef s_n$.
We let $\runs M$ denote the set of all runs of~$M$,
and $\reach M \eqdef \set{\state \sigma \in \Sigma \colon \sigma \in \runs M}$
the set of all \emph{reachable states} of~$M$.
%

\paragraph{Analysis Instances.}
A lattice~$\tup{D, \sqsubseteq_D, \join_D, \meet_D}$ is a poset
with a binary, least upper bound operator $\join_D$ called
\emph{join} and a binary, greatest lower bound operator $\meet_D$ called
\emph{meet}.
%
%
%
%
A~\emph{transformer} $f \colon D \to D$ is a monotone function on 
$D$.
A~\emph{domain} $\tup{D, \sqsubseteq, F}$ consists of a lattice 
and a set of transformers.
We adopt standard assumptions in the literature that~$D$ has a least
element~$\bot$, called \emph{bottom}, and that transformers
are~\emph{bottom-strict},~\ie~$f(\bot) = \bot$.
To simplify presentation, we equip domains with sufficient structure
to lift notions from transition systems to domains,
and assume that domains represent control and data states.

\begin{definition}
An~{analysis instance}~$\ddom \eqdef \tup{D, \sqsubseteq, F, d_0}$, 
consists of a domain $\tup{D, \sqsubseteq, F}$
and an {initial element} $d_0 \in D$.
\label{def:ainst}
\end{definition}
A transformer~$f$ is \emph{enabled} at an element $d$ 
when $f(d) \ne \bot$, and the result of \emph{firing} $f$ at $d$ 
is $f(d)$.
%
%
%
The element~\emph{generated by} or \emph{reached by} a sequence of
transformers $\sigma \eqdef f_1, \ldots, f_m$ is the application 
$\state \sigma \eqdef (f_m \comp \ldots \comp f_1) (d_0)$
of transformers in $\sigma$ to $d_0$.
Let ~$\reach \ddom$ be the set of reachable elements of $\ddom$.
The sequence $\sigma$ is a \emph{run} if $\state \sigma \ne \bot$
and $\runs \ddom$ is the set of all runs of~$\ddom$.

The~\emph{collecting semantics} of a transition system ~$M$ is the analysis 
instance~$\ccdom M \eqdef \tup{\pows \Sigma, \subseteq, F, \set{s_0}}$,
where~$F$ contains a transformer
$
f_a (S) \eqdef \set{s' \in \Sigma \colon s \in S \land s \fire a s'}
$
for every statement $a$ of the program.
The \emph{pointwise-lifting} of a relation $R \subseteq A \times A$ on
statements to transformers in $\ddom$ is 
$R_\ddom = \set{\tup{f_a, f_{a'}} \setsep \tup{a,a'} \in R}$.
Let $m_0 \colon A \to F$ be 
map from statements to transformers: $m_0(a) \eqdef f_a$.
An analysis instance
$\bar \ddom = \tup{\bar D, \barsqsub, \bar F, \bar d_0}$
is an \emph{abstraction} of 
$\tup{D, \sqsubseteq, F, d_0}$
if there exists a \emph{concretization function}
$\gamma: \bar D \to D$, 
which is monotone and satisfies that
$d_0 \sqsubseteq \gamma(\bar d_0) $,
and that 
$f \comp \gamma \sqsubseteq \gamma \comp \bar f$,
where the order between functions is pointwise.

\paragraph{Labelled Prime Event Structures.}
Event structures are tree-like representations of system behaviour
that use partial orders to represent concurrent interaction.
\figref{exunf}~(c) depicts an event structure.
The nodes are events and solid arrows, 
represent causal dependencies: events $4$ and $7$ must fire before $8$
can fire.
The dotted line represents conflicts: $4$ and $7$
are not in conflict and may occur in any order, but $4$ and $9$ are in
conflict and cannot occur in the same execution.

A \emph{labelled prime event structure}~\cite{NPW81}
(\pes)
is a tuple $\les \eqdef \tup{E, <, \cfl, h}$ with
a set of events $E$,
a causality relation ${<} \subseteq E \times E$, which is a strict partial order,
a conflict relation ${\cfl} \subseteq E \times E$ that is symmetric
and irreflexive, and a labelling function $h \colon E \to X$.
The components of $\les$ satisfy
  (1)
  the \emph{axiom of finite causes}, 
  that for all $e \in E$, $\set{e' \in E \colon e' < e}$ is finite, and
  (2)
  the \emph{axiom of hereditary conflict},
  that for all $e,e',e'' \in E$, if $e \cfl e'$ and $e' < e''$, then $e \cfl e''$.
%

The \emph{history} of an event
$\causes e \eqdef \set{e' \in E \colon e' < e}$
is the least set of events that must fire before $e$ can fire.
A \emph{configuration} of $\les$ is a finite set~$C \subseteq E$
that is 
  (\emph{i})
  (causally closed)
  $\causes e \subseteq C$ for all $e \in C$, and
  (\emph{ii})
  (conflict free)
  $\lnot (e \cfl e')$  for all $e, e' \in C$.
We let $\conf \les$ denote the set of all configurations of~$\les$.
For any $e \in E$, the \emph{local configuration} of~$e$ is defined as
$[e] \eqdef \causes e \cup \set e$.
In \figref{exunf}~(c), the set $\set{1,2}$ is a configuration, and in fact it is
a local configuration, \ie, $[2] = \set{1,2}$.
The set~$\set{1,2,3}$ is a $\subseteq$-maximal configuration.
The local configuration of event~8 is $\set{4,7,8}$.
Given a configuration~$C$, we define the \emph{interleavings} of~$C$ as
$\inter C \eqdef
\set{h(e_1), \ldots, h(e_n) \colon
\forall
e_i, e_j \in C, e_i < e_j \implies i < j}$.
An interleaving corresponds to the sequence labelling any
topological sorting (sequentialization) of the events in the configuration.
We say that $\les$ is finite iff $E$ is finite.
\figref{exunf}~(d) shows the interleavings of configuration~$\set{1,2,3}$.

Event structures are naturally (partially) ordered by a \emph{prefix}
relation~$\ispref$.
Given two \peses
$\les \eqdef \tup{E, <, \cfl, h}$ and
$\les' \eqdef \tup{E', <', \cfl', h'}$,
we say that
$\les$ is a \emph{prefix} of~$\les'$, written $\les \ispref \les'$,
when
$E \subseteq E'$,
$<$ and $\cfl$ are the projections of $<'$ and $\cfl'$ to~$E$,
and $E \supseteq \set{e' \in E' \colon e' < e \land e \in E}$.
Moreover, the set of prefixes of a given \pes~$\les$ equipped with~$\ispref$ is a
complete lattice.

\section{Independence for Transformers}
\label{sec:ind}

Partial-order reduction tools use a notion called independence to
avoid exploring concurrent interleavings that lead to the same state.
Our analyzer uses independence between transformers to compactly
represent transformer applications that lead to the same result.
The contribution of this section is a notion of independence for
transformers (represented by the lowest horizontal line in~\cref{fig:overview}) 
and a demonstration that abstraction may both create and
violate independence relationships.

We recall a standard notion of independence for
statements~\cite{Pel93,God96}.
Two statements~$a, a'$ of a program~$M$ \emph{commute at} a state~$s$ iff
\begin{itemize}
\item
  if $a \in \enabl s$ and $s \fire a s'$,
  then $a' \in \enabl s$ iff $a' \in \enabl{s'}$; and
\item
  if $a, a' \in \enabl s$,
  then there is a state $s'$ such that
  $s \fire{a.a'} s'$ and
  $s \fire{a'.a} s'$.
\end{itemize}
Independence between statements is an underapproximation of
commutativity.
A~relation ${\indep} \subseteq A \times A$ is an
\emph{independence} for~$M$ if it is symmetric, irreflexive, and 
satisfies that every 
$(a, a') \in {\indep}$ commute at every reachable state of~$M$.
In~general, $M$ has multiple independence relations; $\emptyset$ is always
one of them.
Suppose independence for transformers is defined by replacing
statements and transitions with
transformers and transformer application, respectively.
\cref{ex:weakindep} illustrates that 
an independence relation on statements cannot be lifted to obtain
transformers that are independent under such a notion.
\begin{example}
\label{ex:weakindep}
Consider the collecting semantics~$\ccdom M$ of a program~$M$ with 
two variables, \texttt{x} and \texttt{y}, 
two statements $a \eqdef \texttt{assume(x==0)}$ and
$a' \eqdef \texttt{assume(y==0)}$, and initial 
element~$d_0 \eqdef \set{\tup{x \mapsto 0, y \mapsto 1}, \tup{x \mapsto 1, y \mapsto 0}}$.
Since $a$ and $a'$ read different variables,
$R \eqdef \set{\tup{a,a'}, \tup{a',a}}$ is an independence relation
on~$M$.
Now observe that
$\set{\tup{f_a,f_{a'}}, \tup{f_{a'},f_a}}$ is not an independence relation
on~$\ccdom M$, as~$f_a$ and~$f_{a'}$ disable each other.
Note, however, that $f_a(f_{a'}(d_0))$ and $f_{a'}(f_{a}(d_0))$ are both
$\bot$.
\end{example}

Weak independence, defined below, allows transformers to be
considered independent even if they disable each other.

\begin{definition}
Let $\ddom \eqdef \tup{D, \sqsubseteq, F, d_0}$
be an analysis instance.
A relation ${\indep} \subseteq F \times F$ is a {weak independence} on
transformers if it is symmetric, irreflexive, and satisfies that 
$f \indep f'$ implies $f(f'(d)) = f'(f (d))$
for every $d \in \reach \ddom$.
Moreover, $\indep$ is an \emph{independence} if it is a weak
independence and satisfies that if 
$f(d) \neq \bot $, then 
$(f \comp f') (d) \neq \bot$ iff $f'(d) \neq \bot$,
for all $d \in\reach \ddom$.
\label{def:ind-trans}
\end{definition}
Recall that $R_\ddom$ is the lifting of a relation on statements to
transformers.
Observe that the relation $R$ in~\cref{ex:weakindep}, when lifted to
transformers is a weak independence on~$\ccdom M$.
The proposition below shows that independence relations on
statements generate weak independence on transformers over~$\ccdom M$.
\begin{restatable}[Lifted independence]{proposition}{liftedindep}
If $\indep$ is an independence relation on~$M$,
the lifted relation~$\indep_{\ccdom M}$ is a weak independence 
on the collecting semantics~$\ccdom M$.
\label{prop:lifting-independence}
\end{restatable}
We now show that independence and abstraction are distinct notions in
that transformers that are independent in a concrete domain may
not be independent in the abstract, and those that are not independent
in the concrete may become independent in the abstract.

Consider an analysis instance $\bar\ddom \eqdef \tup{\bar D, \barsqsub, \bar F, \bar d_0}$
that is an abstraction of $\ddom \eqdef \tup{D, \sqsubseteq, F, d_0}$
and a weak independence ${\indep} \subseteq F \times F$.
The \emph{inherited relation} 
$\bar \indep \subseteq \bar F \times \bar F$
contains  $\tup{\bar f, \bar f'}$ iff $\tup{f, f'}$ is in $\indep$.

\begin{example}[Abstraction breaks independence]
\label{ex:abstraction.breaks}
Consider a system~$M$ 
with the initial state $\tup{x \mapsto 0,y \mapsto 0}$,
and two threads $t_1: \texttt{x = 2}$, $t_2: \texttt{y = 7}$.
Let $\idom$ be the domain for interval analysis
with elements $\tup{i_x, i_y}$ being intervals for values of $x$ and
$y$.
The initial state is~$\bar d_0 = \tup{x \mapsto [0,0], y \mapsto [0,0]}$.
Abstract transformers for $t_1$ and $t_2$ are shown below.
These transformers are deliberately imprecise to highlight that
sound transformers are not the most precise ones.
\begin{align*}
f_1(\tup{i_x,i_y}) & = \tup{[2,4], i_y}
 & 
f_2(i_x,i_y) = \tup{i_x, ~(\text{if~} 3 \in i_x \text{ then } [7,9] \text{ else } [6,8])}
\end{align*}
The relation ${\indep} \eqdef \set{(t_1, t_2), (t_2, t_1)}$ is an
independence on~$M$, and when lifted to $\indep_{\ccdom M}$ is a weak
independence on $\ccdom M$  (in fact, $\indep_{\ccdom M}$ is an independence).
However, the relation $\indep_{\idom}$ is not a weak independence
because $f_1$ and~$f_2$ do not commute at~$d_0$, 
due to the imprecision introduced by abstraction.
Consider the statements ~$\texttt{assume(x != 9)}$
and~$\texttt{assume(x < 10)}$ applied to ~$\tup{x \mapsto [0,10]}$ to
see that even best transformers may not commute.
\end{example}

On the other hand, even when certain transitions are not independent, their
transformers may become independent in an abstract domain.

\begin{example}[Abstraction creates independence]
\label{ex:abstraction.creates}
Consider two threads $t_1: \texttt{x = 2}$ and $t_2: \texttt{x = 3}$, 
with abstract transformers $f_1(i_x)  = [2,3]$ and $f_2(i_x) = [2,3]$.
The transitions $t_1$ and $t_2$ do not commute, but owing to
imprecision, $R =\allowbreak \set{(f_1, f_2),\allowbreak (f_2, f_1)}$ is a weak independence
on~$\idom$.
\end{example}

\section{Unfolding of an Abstract Domain with Independence}
\label{sec:unf}

This section shows that unfoldings, which have primarily been used to
analyze Petri nets, can be applied to abstract interpretation (represented
by vertical lines in~\cref{fig:overview}).
An abstract unfolding is an event structure in which an event is recursively
defined as the application of a transformer after a minimal set of interferring
events; and a configuration represent equivalent sequences of transformer
applications (events).
Analogous to an invariant map in abstract interpreters and an abstract
reachability tree in software model checkers, our abstract unfolding
allows for constructing an over-approximation of the set of firable
transitions in a program.

\subsection{The Unfolding of a Domain}
\label{sec:unf.def}

Our construction generates a PES 
$\les \eqdef \tup{E, <, {\cfl}, h}$.
Recall that a configuration is a set of events that is closed with
respect to $<$ and that is conflict-free.
Events in $\les$ have the form $e = \tup{f,C}$,
representing that the transformer $f$ is applied after the
transformers in configuration $C$ are applied.
The order in which transformers must be applied is given by~$<$,
while $\cfl$ encodes transformer applications that cannot belong to the same
configuration.

The unfolding $\unf{\ddom,\indepx}$ of an analysis instance
$\ddom \eqdef \tup{D, \sqsubseteq, F, d_0}$
with respect to a relation~${\indepx} \subseteq F \times F$
is defined inductively below.
Recall that a configuration $C$ generates a set of interleavings
$\inter C$, which define the \emph{state} of the configuration.
\begin{align*}
\state C & \eqdef \bigsqcap_{\sigma \in \inter C} \state \sigma
\end{align*}
If $\indepx$ is a weak independence relation, all interleavings 
lead to the same state.
\begin{definition}[Unfolding]
\label{def:unf}
The {unfolding} $\unf{\ddom,\indepx}$ of~$\ddom$ under the
relation~$\indepx$ is the structure returned by the following procedure:
\begin{enumerate}
\item
  Start with a \pes $\les \eqdef \tup{E, <, {\cfl}, h}$
  equal to $\tup{\emptyset, \emptyset, \emptyset, \emptyset}$.
\item
 Add a new event $e \eqdef \tup{f,C}$ to~$E$, where
   the configuration $C \in \conf \les$ and transformer $f$ satisfy
   that $f$ is enabled at $\state C$, and
    $\lnot (f \indepx h(e))$ holds for every $<$-maximal event $e$ in~$C$.
\item
  Update $<$, $\cfl$, and $h$ as follows:
  \begin{itemize}
  \item
    for every $e' \in C$, set $e' < e$;
  \item
    for every $e' \in E \setminus C$, if $e \ne e'$ and $\lnot (f \indepx h(e'))$,
    then set $e' \cfl e$;
  \item
    set $h(e) \eqdef f$.
  \end{itemize}
\item
  Repeat steps 2 and 3 until no new event can be added to~$E$;
  return $\les$.
\end{enumerate}
\end{definition}

\cref{def:unf} 
defines the events, the causality, and conflict relations
of $\unf{\ddom,\indepx}$ by means of a saturation procedure.
Step 1 creates an empty \pes.
Step 2 defines a new event from a transformer~$f$ that can be applied
after configuration~$C$.
Step~3 defines~$e$ to be a causal successor of every
dependent event in~$C$, and defines $e$ to be in conflict 
with dependent events not in $C$.
Since conflicts are inherited in a \pes, causal successors of $e$ will
also be in conflict with all $e'$ satisfying $e \cfl e'$.
Events from $E \setminus C$, which are unrelated to~$f$ in~$\indepx$,
will remain concurrent to~$e$.

\begin{restatable}{proposition}{unfunique}
The structure $\unf{\ddom,\indep}$ generated by \cref{def:unf} is a uniquely
defined~\pes.
\label{prop:unf-unique}
\end{restatable}

%

If $\indepx$ is a weak independence,
every configuration of $\unf{\ddom,\indepx}$ represents sequences of
transformer applications \emph{that produce the same element}.
If $C$ is a configuration that is local,
meaning it has a unique maximal event, or if $C$ is generated by an
independence, then $\state C$ will not be $\bot$.
Treating transformers as independent if they generate $\bot$ enables
greater reduction during analysis.
\begin{restatable}[Well-formedness of $\unf{\ddom,\indep}$]{theorem}{unfsound}
\label{thm:unf.sound}
Let $\indep$ be a \emph{weak} independence on~$\ddom$,
let $C$ be a configuration of $\unf{\ddom,\indep}$ 
and $\sigma, \sigma'$ be interleavings of~$C$. Then:
\begin{enumerate}
\item
  $\state \sigma = \state{\sigma'}$;
\item
  $\state \sigma \ne \bot$ when
  $\indep$ is additionally an independence relation;
\item
  If $C$ is a local configuration, then also $\state \sigma \ne \bot$.
\end{enumerate}
\end{restatable}

\cref{thm:unf.complete} shows that the unfolding is adequate for
analysis in the sense that every sequence of transformer applications
leading to non-$\bot$ elements that could be generated during standard
analysis with a domain will be contained in the unfolding.
We emphasize that these sequences are only symbolically represented.

\begin{restatable}[Adequacy of $\unf{\ddom,\indep}$]{theorem}{unfcomplete}
\label{thm:unf.complete}
For every {weak} independence relation $\indep$ on~$\ddom$,
and sequence of transformers $\sigma \in \runs \ddom$,
there is a unique configuration~$C$ of
$\unf{\ddom,\indep}$ such that $\sigma \in \inter C$.
\end{restatable}

\subsection{Abstract Unfoldings}
\label{sec:abs-unf}

The soundness theorems of abstract interpretation show when a fixed
point computed in an abstract domain soundly approximates fixed points
in a concrete domain.
Our analysis constructs unfoldings instead of fixed points.
The soundness of our analysis \emph{does not} follow
from fixed point soundness because the abstract unfolding we construct
depends on the independence relation used.
Though independence may not be preserved under lifting,
as shown in \cref{ex:abstraction.breaks},
lifted relations can still be used to obtain sound results.

\begin{example}
In \exref{abstraction.breaks},
the transformer composition $f_1 \comp f_2$
produces $\tup{x \mapsto [2,4],\allowbreak y \mapsto [6,8]}$,
while $f_2 \circ f_2$ produces
$\tup{x \mapsto [2,4],\allowbreak y \mapsto [7,9]}$.
If $f_1$ and $f_2$ are considered independent, 
the state of the configuration $\set{f_1, f_2}$
is 
$\state {f_1,f_2}\meet \state {f_2, f_1}$,
which is the abstract element
$\tup{x \mapsto [2,4], y \mapsto [7,7]}$
and contains the final state
$\tup{x \mapsto 2, y \mapsto 7}$ reached in the concrete. 
\label{ex:soundness.with.nonindependence}
\end{example}

Thus, with sound abstractions of (weakly) independent, concrete transformers,
can be treated as independent without compromising soundness of the
analysis. 
The soundness theorem below asserts a correspondence between sequences
of concrete transformer applications and the abstract unfolding.
The concrete and abstract objects in \cref{thm:absunf.correct} have
different type: we are not relating a concrete unfolding with an
abstract unfolding, but concrete transformer sequences with abstract
configurations.
Since $\state C$ is defined as a meet of transformer sequences,
the proof of \cref{thm:absunf.correct} relies on the
independence relation
and has a different structure 
from standard proofs of fixed point soundness from transformer
soundness.

\begin{restatable}[Soundness of the abstraction]{theorem}{absunfcorrect}
\label{thm:absunf.correct}
Let $\bar \ddom$ be a sound abstraction  of the analysis instance $\ddom$, 
let $\indep$ be a weak independence on $\ddom$,
and $\bar \indep$ be the lifted relation on $\bar \ddom$.
For every sequence $\sigma \in \runs \ddom$ 
satisfying $\state \sigma \neq \bot$,
there is a unique configuration~$C$ of
$\unf{\bar\ddom,\barindep}$ such that $m(\sigma) \in \inter C$.
\end{restatable}

\thmref{absunf.correct} and \thmref{unf.complete} are fundamentally
different.
\thmref{unf.complete} shows that an unfolding parameterized by a weak
independence relation is a data structure for representing all sequences
of transformer applications that may be generated during analysis
within a domain.
\thmref{absunf.correct} shows that every concrete sequence of transformers
has a corresponding sequence of abstract transformers. 
However, the abstract unfolding in \thmref{absunf.correct} may not
represent all transformer applications of the abstract domain in
isolation.
Formally, 
let $\ispref$ be the order between unfolding prefixes and
$m(\unf{\ddom,\indep})$ is a lifting of an unfolding over a concrete
domain to an abstract domain, 
we have
$m(\unf{\ddom,\indep}) ~ \ispref ~\; \unf{\bar\ddom,\barindep}$.
In fact, every configuration of~$\unf{\ddom,\indep}$
will be isomorphic to a configuration in~$\unf{\bar\ddom,\barindep}$.
%
\section{Plugging Thread-Local Analysis}
\label{sec:collapse}

Unfoldings compactly represent concurrent executions using partial orders.
However, they are a branching structure and one extension of the 
unfolding can multiply the number of branches, leading to 
a blow-up in the number of branches. 
Static analyses of sequential programs often avoid this explosion
(at the expense of precision) by over-approximating (using join or 
widening) the abstract state at the CFG locations where two or 
more program paths converge.
Adequately lifting this simple idea of merging at CFG locations
from sequential to concurrent programs is a highly non-trivial problem~\cite{FM07}.

In this section, we present a method that addresses this challenge
and can mitigate the blow-up in the size of the unfolding caused
by conflicts between events of the same thread.
The key idea of our method is to merge abstract states
generated by statements that work on local data of one thread, \ie, those 
whose impact over the memory/environment is
invisible to other threads.
Intuitively, the key insight is that we can 
merge certain configurations of the unfolding and 
still preserve its structural properties 
with respect to interference.
The state of the resulting configuration will be
a sound over-approximation of the states of the
merged configurations at no loss of precision
with respect to conflicts between events of 
different threads.

Our approach is to analyse~$M$ by constructing the 
unfolding of an abstract domain~$\ddom$ and a
weak independence relation~${\indep}$
using a thread-local procedure 
that over-approximates the effect of transformers altering local variables.

Assume that~$M$ has~$n$ threads.
Let $F_1, \ldots, F_n$ be the partitioning of the set of transformers~$F$ by the
thread to which they belong.
For $f \in F_i$, we let~$p(f) \eqdef i$ denote the thread
to which~$f$ belongs.
We define, per thread, the (local) transformers which can
be used to run the merging analysis.
%
A~transformer~$f \in F_i$ is \emph{local} when, for all other threads~$j \ne i$
and all transformers $f' \in F_j$ we have $f \indep f'$.
A~transformer is \emph{global} if it is not local.
We denote by~$F_i^\text{loc}$ and~$F_i^\text{glo}$, respectively,
the set of local and global transformers in~$F_i$.
In~\cref{fig:forfor}~(a), the global transformers would be those
representing the actions to the variable~\verb!g!.
The remaining statements correspond to local transformers.

We formalize the thread-local analysis using the function
$\collapsee{} \colon \nat \times D \to D$, which plays the 
role of an off-the-shelf static analyzer for sequential
thread code.
A~call to $\collapse{i, d}$ will
run a static analyzer on thread~$i$, restricted to~$F_i^\text{loc}$,
starting from~$d$, and return its result which we assume is a sound fixed point.
Formally, we assume that $\collapse{i,d}$ returns $d' \in D$,
such that for every sequence~$f_1 \ldots f_n \in (F_i^\text{loc})^*$ 
we have~$(f_n \comp \ldots \comp f_1)(d) \sle d'$.
This condition requires any implementation of~$\collapse{i,d}$ to return a sound
approximation of the state that thread~$i$ could possibly reach
after running only local transformers starting from~$d$.

\SetKwProg{Proc}{Procedure}{}{}
\SetKwIF{If}{ElseIf}{Else}{if}{}{else if}{else}{end}
\SetKwFor{ForEach}{foreach}{}{end}
\SetKwFor{ForAll}{forall}{}{end}
\SetKwFor{For}{for}{}{end}
\SetInd{-0.1em}{0.7em}
\SetNlSty{}{\color{gray}}{}
\SetVlineSkip{0em}
\SetKw{Continue}{continue}
\SetKw{In}{in}
\SetKwRepeat{Do}{do}{while}

\SetKwFunction{unfold}{unfold}
\SetKwFunction{iscutoff}{iscutoff}
\SetKwFunction{ccollapse}{tla}
\SetKwFunction{mkevent}{mkevent}

\begin{algorithm}[t]
\DontPrintSemicolon
\setstretch{1.1}

\begin{multicols}{2}
\Proc {\unfold {$\ddom, {\indep}, n$}}
{
   Set $\les \eqdef \tup{E,<,{\cfl},h}$ to $\tup{\emptyset,\emptyset,\emptyset,\emptyset}$
   
   \ForAll {$i,C$ \In $\nat_n \times \conf \les$}
   {
      \For {$f$ enabled on $\ccollapse {$i, \state C$}$}
      {
         $e \eqdef \mkevent{$f, C, {\indep}$}$\\
         \lIf {\iscutoff {$e, \les$}} {\Continue}
         Add $e$ to $E$\\
         Extend $<$, $\cfl$, and $h$ with $e$.
      }
   }
}

\Proc {\mkevent {$f, C, {\indep}$}}
{
   \Do
   {$C$ changed   
   } 
   {
   Remove from $C$ any $<$-maximal event $e$ such that
   $f \indep h(e)$
   }
   \Return $\tup{f,C}$
}
\end{multicols}

\medskip
\caption{Unfolding using thread-local fixpoint analysis}
\label{a:a1}
\end{algorithm}

\cref{a:a1} presents the overall approach proposed in this paper.
Procedure~\unfold builds an~\emph{abstract unfolding} for~$\ddom$ under
independence relation~$\indep$.
It non-deterministically selects a thread~$i$ and a configuration~$C$ and runs
a sequential static analyzer on thread~$i$ starting on the state reached by~$C$.
If a global transformer~$f \in F_i^\text{glo}$ is eventually enabled, the
algorithm will try to insert it into the unfolding.
For that it first calls the function~\mkevent that will generate a
an event, \ie an history for~$f$ from~$C$ according to~\cref{def:unf}.
%
%
If the new event~$e$ is a \emph{cutoff}, \ie an~\emph{equivalent} event
is already in the unfolding prefix, then it will be ignored.
Otherwise, we add it to~$E$.
Finally, we update relations $<$, $\cfl$, and $h$ using exactly the same
procedure as in Step~3 of~\cref{def:unf}.

We denote by~$\preff{\ddom,\indep}$ the PES constructed by a call
to~\unfold{$\ddom,\indep,n$}.
Events of $\preff{\ddom,\indep}$ are labelled by global transformers of~$\ddom$.
As a result, we adapt the definition of~$\state C$ to
account for the effects of~\ccollapse on a thread.
See~\cref{sec:app.collapse}.
When the~\ccollapse performs a path-insensitive analysis, the structure~$\pref{\ddom,\indep}$ is
(\emph{i})~path-insensitive for runs that execute only local code,
(\emph{ii})~partially path-sensitive for runs that execute one or more
global transformer, and
(\emph{iii})~flow-sensitive with respect to interference between threads.
We refer to this analysis as a~\textit{causally-sensitive} analysis as 
it is precise with respect to the \emph{dynamic} interference between threads.

\cref{a:a1} embeds multiple constructions explained in this paper.
For instance, when $\collapsee$ is implemented by the function $g(d,i) \eqdef d$
and the check of cutoffs is disabled (\iscutoff systematically returns
\textit{false}), the algorithm is equivalent to~\cref{def:unf}.
We now show that~$\preff{\ddom,\indep}$ is a safe abstraction of~$\ddom$
when~\ccollapse performs a non-trivial operation.

\begin{restatable}[Soundness of the abstraction]{theorem}{collapsecorrect}
\label{thm:collapse.correct}
Let $\indep$ be a weak independence on~$\ddom$ and
$\pref{\ddom,\indep}$ the PES computed by a call to
\unfold{$\ddom, \indep, n$} with cutoff checking disabled.
Then, for any execution $\sigma \in \runs \ddom$
there is a unique configuration~$C$ in
$\pref{\ddom,\indep}$ such that $\hat\sigma \in \inter C$.
\end{restatable}

\subsection{Cutoff Events: Pruning the Unfolding}
\label{sec:cutoffs}

If we remove the conditional statement in line~6 of~\cref{a:a1},
the algorithm would only terminate if every run of~$\ddom$ contains 
finitely many global transformers.
This conditional check has two purposes:
(1) preventing infinite executions from inserting infinitely many events
into~$\les$;
(2) pruning branches of the unfolding that start with~\emph{equivalent} events.
The procedure~\iscutoff decides when an event is marked as a~\emph{cutoff}~\cite{Mcm93}.
In such cases, no causal successor of the event will be explored.
The implementation of~\iscutoff cannot prune ``too often'', as we want the
computed PES to be a~\emph{complete} representation of behaviours of~$\ddom$
(\eg, if a transformer is fireable, then some event in the PES will be labelled
by it).

Formally, given~$\ddom$, a PES $\les$ is $\ddom$-\emph{complete} iff
for every reachable element $d \in \reach \ddom$ there is a configuration~$C$
of~$\les$ such that $\state C \sge d$.
The key idea behind cutoff events is that,
if event~$e$ is marked as a cutoff, then for any configuration~$C$ that
includes~$e$ it must be possible to find a configuration~$C'$ without cutoff
events such that $\state C \sle \state{C'}$. 
This can be achieved by defining~\iscutoff{$e, \les$} to be the predicate:
$
\exists e' \in \les \text{ such that } 
\state{[e]} \sle
\state{[e']} \text{ and }
|[e']| < |[e]|.
$
When such~$e'$ exists, including the event~$e$ in~$\les$ is unnecessary
because any configuration~$C$ such that~$e \in C$ can be replayed in~$\les$
by first executing~$[e']$ and then (copies of) the events in~$C \setminus [e]$.

We now would like to prove that \cref{a:a1} produces a $\ddom$-complete prefix
when instantiated with the above definition of~\iscutoff.
However, a subtle an unexpected interaction between the operators~\ccollapse
and~\iscutoff makes it possible to prove \cref{thm:a1complete} only when
\ccollapse \emph{respects independence}.
Formally, we require~\ccollapse to satisfy the following property:
for any $d \in \reach \ddom$ and any two global transformers
$f \in F_i^\text{glo}$ and
$f' \in F_j^\text{glo}$,
if $f \indep f'$ then
\[
(f' \comp \collapse{j} \comp
f \comp \collapse{i}) (d) =
(f \comp \collapse{i} \comp
f' \comp \collapse{j}) (d)
\]
When \ccollapse does not respect independence, it may
over-approximate the global state (\eg via joins and widening) in 
a way that breaks the independence of otherwise independent global transformers.
This triggers the cutoff predicate to incorrectly prune necessary events.

\begin{restatable}{theorem}{algoprefixcomplete}
\label{thm:a1complete}
Let $\indep$ be a weak independence in~$\ddom$.
Assume that \ccollapse respects independence and that
\iscutoff uses the procedure defined above.
Then the PES $\preff{\ddom,\indep}$ computed by \cref{a:a1} is $\ddom$-complete.
\end{restatable}

Note that~\cref{a:a1} terminates if the
lattice order $\sle$ is a well partial order
(every infinite sequence contains an increasing pair). This
includes, for instance, all finite domains.
Furthermore, it is also possible to accelerate the 
termination of~\cref{a:a1} using widenings in~\ccollapse
to~\emph{force cutoffs}.
Finally, notice that while we defined~\iscutoff using McMillan's size
order~\cite{Mcm93}, \cref{thm:a1complete} also holds if~\iscutoff is defined
using adequate orders~\cite{ERV02}, known to yield smaller prefixes.

\section{Experimental Evaluation}
\label{sec:exp}
In this section we evaluate our approach based on abstract unfoldings.
The goal of our experimental evaluation is to explore the following questions:
\begin{itemize}
	\item Are abstract unfoldings practical? (I.e., is our approach
	able to yield efficient algorithms that can be used to prove 
	properties of concurrent programs that require precise interference
	reasoning?)
	\item How does abstract unfoldings compare with competing 
	approaches such as thread-modular analysis and 
	symbolic partial order reduction?
\end{itemize} 

\paragraph{Implementation.} 
\label{sub:implementation}
To address these questions, we have implemented a new program analyser
based on abstract unfoldings baptized~\apoet,
which implements an efficient variant of the exploration 
algorithm described in~\cref{a:a1}.
The exploration strategy is based on~\poet~\cite{RSSK15}, 
an explicit-state model checker that implements a super-optimal 
partial order reduction method using unfoldings.

As described in~\cref{a:a1},~$\apoet$ is an analyser 
parameterized by a domain and a set of procedures:~\ccollapse,
~\iscutoff and~\mkevent.
As a proof of concept, we have implemented an interval analysis and a basic
parametric segmentation functor for arrays~\cite{CCL11}, which we
instantiate with intervals and concrete integers values (to represent
offsets).
In this way, we are able to precisely handle arrays of 
threads and mutexes.
$\apoet$ supports dynamic thread creation and uses~$\cil$ to inline 
functions calls.
The analyser receives as input a concurrent C program 
that uses the POSIX thread library and parameters to control 
the widening level and the use of cutoffs.
We implemented cutoffs according to the definition in~\cref{sec:cutoffs}
using an hash table that maps control locations to abstract values and 
the size of the local configuration of events.

$\apoet$ is parameterized by a domain functor of actions 
that is used to define independence and control the~\ccollapse procedure.
We have implemented an instance of the domain of actions for
memory accesses and thread synchronisations.
Transformers~\emph{record} the segments of the memory, intervals of addresses
or sets of addresses, that have been read or written and synchronisation actions 
related to thread creation, join and mutex lock and unlock operations.
This approach is used to compute a conditional independence relation
as transformers can perform different actions depending on the state.
The conditional independence relation is dynamically computed and 
is used in the procedure~\mkevent.

Finally, the~\ccollapse procedure was implemented
with a worklist fixpoint algorithm which uses the widening
level given as input.
In the interval analysis, we guarantee that~\ccollapse respects
independence using a predicate over the actions that identifies
whether a transformer is local or global.
This modularity allows us to define two modes of analysis for~$\apoet$:
1)~consider global transformers those that yield actions related to thread
synchronisation (\ie, thread creation/join and mutex lock/unlock) assuming
that the program is data-race free and 2)~consider an action global if it
accesses the heap or is related to thread synchronisation which can be used
to detect data races.
\dk{The sentence above needs to be re-written, please!}


\paragraph{Benchmarks.} 
\label{sub:experiments}
We employ 6 benchmarks adapted from the SVCOMP'17 (yielding 9
rows in \cref{table:results})
and 4 parametric programs (yielding 15 rows) written by us: map-reduce DNA
sequence analysis, producer-consumer, parallel sorting, and a thread pool.
Most SVCOMP benchmarks are unsuitable for this comparison because
either they are data-deterministic (and our approach fights data-explosion)
or create unboundedly many threads, or use non-integer data types (\eg, structs,
unsupported by our prototype).
Thus we use new benchmarks exposing data non-determinism and complex
synchronization patterns, where the correctness of assertions depend on the
history of synchronizations.
All new benchmarks are as complex as the most complex ones of the SVCOMP
(excluding device drivers).

Each program was annotated with assertions enforcing, among others,
properties related to thread synchronisation (\eg, after spawning the worker
threads, the master analyses results only after all workers finished),
or invariants about data (\eg, each thread accesses a non-overlapping segment of
the input array).

\paragraph{Tools compared.}
We compare $\apoet$ against the two approaches most closely related to ours:
abstract interpreters (represented by the tool~\astreea) and partial-order
reductions (PORs) handling data-nondeterminism (tools~\impara
and~\cbmc~5.6).
\astreea implements thread-modular abstract interpretation for concurrent
programs~\cite{Mine14},
\impara combines POR with interpolation-based reasoning to cope with data
non-determinism~\cite{WKO13}, and
\cbmc uses a symbolic encoding based on partial orders~\cite{AKT13}.
We sought to compare against symbolic execution tools but we could not find any
available to download or capable of parsing the benchmarks.

\newcommand\newrow{\\[-3pt]}
\newcommand\param[1]{\footnotesize(#1)}

\begin{table}[!t]

\footnotesize

\caption{Experimental results.
All experiments with \apoet,~\impara~and~\cbmc~were performed on an Intel Xeon CPU 
with 2.4 GHz and 4 GB memory with a timeout of 30 minutes;~\astreea~was ran 
on HP ZBook with 2.7 GHz i7 processor and 32 GB memory.
Columns are:
$P$:              \nr of threads;
$A$:              \nr of assertions;
$t(s)$:           running time (TO - timeout);
$E$:              \nr of events in the unfolding;
$E_{\text{cut}}$: \nr of cutoff events;
$W$:              \nr of warnings;
$V$:              verification result (S - safe; U - unsafe);
$N$:              \nr of node states;
A $*$ marks programs containing bugs.
{\tt <2} reads as ``\emph{less than 2}''.
}
\label{table:results}

\setlength\tabcolsep{3pt}
\def\sep{\hspace{10pt}}
\def\tinysep{\hspace{4pt}}

\centering
\tt
\begin{tabular}{lrr@{\tinysep}rrrrr@{\tinysep}cr@{\tinysep}rrr@{\tinysep}rr}
\toprule
  \multicolumn{3}{l}{\rm \footnotesize Benchmark}
& \multicolumn{4}{l}{\rm \footnotesize \apoet}
& \multicolumn{2}{l}{\rm \scriptsize \astreea}
& \multicolumn{3}{l}{\rm \footnotesize \impara}
& \multicolumn{2}{l}{\rm \footnotesize \cbmc\ 5.6}
\\[-2pt]
  \cmidrule(r){1-3}
  \cmidrule(r){4-7}
  \cmidrule(r){8-9}
  \cmidrule(r){10-12}
  \cmidrule(r){13-14}
  \rm Name 
& $P$
& $A$
& $t(s)$
& $E$
& $E_{\text{cut}}$
& $W$
& $t(s)$
& $W$
& $V$
& $t(s)$
& $N$
& $V$
& $t(s)$
\\[-2pt]
\midrule
\rm\sc {atgc\param{2}}      & 3 &  7 &    0.37  &    47 &     0  &  1 & 1.07 & 2 & - &    TO &    - & S &   2.37 \newrow
\rm\sc {atgc\param{3}}      & 4 &  7 &    5.78  &   432 &     0  &  1 & 1.69 & 2 & - &    TO &    - & S &    6.6 \newrow
\rm\sc {atgc\param{4}}      & 5 &  7 &  132.08  &  7195 &     0  &  1 & 2.68 & 2 & - &    TO &    - & S &  20.22 \newrow
\rm\sc {cond}               & 5 &  2 &    0.55  &   982 &     0  &  2 & 0.71 & 2 & - &    TO &    - & S &  34.39 \newrow 
\rm\sc {fmax\param{2,3}}    & 2 &  8 &    0.70  &   100 &    15  &  0 & 0.31 & 0 & - &    TO &    - & - &     TO \newrow
\rm\sc {fmax\param{3,3}}    & 2 &  8 &    0.58  &    85 &    11  &  0 &   <2 & 2 & - &    TO &    - & - &     TO \newrow
\rm\sc {fmax\param{5,3}}    & 2 &  8 &    0.56  &    85 &    11  &  0 & 1.50 & 2 & - &    TO &    - & - &     TO \newrow
\rm\sc {fmax\param{2,4}}    & 2 &  8 &    3.38  &   277 &    43  &  0 &   <2 & 2 & - &    TO &    - & - &     TO \newrow
\rm\sc {fmax\param{2,6}}    & 2 &  8 &   45.82  &  1663 &   321  &  0 &   <2 & 2 & - &    TO &    - & - &     TO \newrow
\rm\sc {fmax\param{4,6}}    & 2 &  8 &   61.32  &  2230 &   207  &  0 &   <2 & 2 & - &    TO &    - & - &     TO \newrow
\rm\sc {fmax\param{2,7}}    & 2 &  8 &  146.19  &  3709 &   769  &  0 & 1.87 & 2 & - &    TO &    - & - &     TO \newrow
\rm\sc {fmax\param{4,7}}    & 2 &  8 &  285.23  &  6966 &   671  &  0 &   <2 & 2 & - &    TO &    - & - &     TO \newrow
\rm\sc {lazy}               & 4 &  2 &    0.01  &    72 &     0  &  0 & 0.50 & 2 & - &    TO &    - & S &   3.59 \newrow
\rm\sc {lazy*}              & 4 &  2 &    0.01  &    72 &     0  &  1 & 0.49 & 2 & - &    TO &    - & U &   3.50 \newrow
\rm\sc {monab1}             & 5 &  1 &    0.27  &   982 &     0  &  0 & 0.61 & 0 & - &    TO &    - & S &  38.51 \newrow 
\rm\sc {monab2}             & 5 &  1 &    0.25  &   982 &     0  &  0 & 0.58 & 1 & - &    TO &    - & S &  37.34 \newrow 
\rm\sc {rand}               & 5 &  1 &    0.40  &   657 &     0  &  0 & 3.32 & 0 & - &    TO &    - & - &     TO \newrow 
\rm\sc {sigma}              & 5 &  5 &    2.62  &  7126 &     0  &  0 & 0.43 & 0 & - &    TO &    - & S & 189.09 \newrow
\rm\sc {sigma*}             & 5 &  5 &    2.64  &  7126 &     0  &  1 & 0.43 & 1 & - &    TO &    - & U & 141.35 \newrow
\rm\sc {stf}                & 3 &  2 &    0.01  &    69 &     0  &  0 & 0.66 & 2 & S &  5.93 &  250 & S &   2.12 \newrow
\rm\sc {tpoll\param{2}*}    & 3 & 11 &    1.23  &   141 &     7  &  1 & 1.97 & 2 & U &  0.64 &   80 & - &     TO \newrow
\rm\sc {tpoll\param{3}*}    & 4 & 11 &  109.22  &  1712 &    90  &  2 & 3.77 & 3 & U &  0.72 &  113 & - &     TO \newrow
\rm\sc {tpoll\param{4}*}    & 5 & 11 & 1111.46  & 33018 &  1762  &  2 & 8.06 & 3 & U &  0.78 &  152 & - &     TO \newrow
\rm\sc {thpool}             & 2 & 24 &   33.47  &   353 &   103  &  0 & 1.44 & 5 & S &    TO &    - & - &     TO \newrow
\\[-3pt]
\bottomrule
\end{tabular}
\end{table}

\paragraph{Analysis.}
\Cref{table:results} presents the experimental results.
When the program contained non-terminating executions (\eg, spinlocks), we used
5 loop unwindings for \cbmc as well as cutoffs in \apoet and a widening level
of~15.
For the family of~$\textsc{fmax}$ benchmarks, we were not able to
run~\astreea~on all instances, so we report approximated execution times and
warnings based on the results provided by~Antoine Miné on some of the
instances.
With respect to the size of the abstract unfolding, our experiments
show that~\apoet~is able to explore unfoldings up to 33K events and 
it was able to terminate on all benchmarks with an average execution
time of~81 seconds.
In comparison with~\astreea,~\apoet~is far more precise: we obtain only 12
warnings~(of which 5 are false positives) with~\apoet compared to~43~(32
false positives) with~\astreea.
We observe a similar trend when comparing~\apoet~with 
the~\mthread~plugin for~\framac~\cite{YaB12} and confirm that 
the main reason for the source of imprecision 
in~\astreea~is imprecise reasoning of thread interference.
In the case of~\apoet, we obtain warnings in benchmarks that are buggy
($\textsc{lazy*},\textsc{sigma*}$ and $\textsc{tpoll*}$ family), as
expected.
Furthermore,~\apoet reports warnings in the~$\textsc{atgc}$ benchmarks caused
by imprecise reasoning of arrays combined with widening and also in 
the~$\textsc{cond}$ benchmark as it contains non-relational assertions.

\apoet~is able to outperform \impara~and~\cbmc on all benchmarks.
We~believe that these experiments demonstrate that effective symbolic
reasoning with partial orders is challenging as~\cbmc only terminates
on~46\% of the benchmarks and~\impara only on~17\%.


\section{Related Work}
\label{sec:related}

In this section, we compare our approach with closely related program
analysis techniques for (i)~concurrent programs with (ii)~a bounded number
of threads and that (iii)~handle data non-determinism.

The thread-modular approach in the style of rely-guarantee 
reasoning has been extensively studied in 
the past~\cite{Mine14,Mine12,CH09,MPR07,FQ03,KW16,MM17}.
%
In~\cite{Mine14}, Miné proposes a flow-insensitive 
thread-modular analysis based on the interleaving semantics
which forces the abstraction to cope with interleaving explosion.
We address the interleaving explosion using the unfolding
as an algorithmic approach to compute a flow and path-sensitive
thread interference analysis.
A recent approach~\cite{MM17} uses relational domains and 
trace partitioning to recover precision in thread modular
analysis but requires manual annotations to guide the partitioning
and does not scale with the number of global variables.
The analysis in~\cite{FK12} is not as precise as our approach 
(confirmed by experiments with~\duet~on a simpler version of 
our benchmarks) as it employs an abstraction for unbounded parallelism.
The work in~\cite{KW16} presents a thread modular analysis that 
uses a lightweight interference analysis to achieve an higher level
of flow sensitivity similar to~\cite{FK12}.
The interference analysis of~\cite{KW16} uses a constraint system
to discard unfeasible pairs of read-write actions which is static 
and less precise than our approach based on independence.
The approach is also flow-insensitive 
in the presence of loops with global read operations.

The interprocedural analysis for recursive concurrent programs
of~\cite{Jeannet12} does not address the interleaving explosion.
A related approach that uses unfoldings is the causality-based 
bitvector dataflow analysis proposed in~\cite{FM07}.
There, unfoldings are used as a method to obtain 
dataflow information while in our approach they are the 
fundamental datastructure to drive the analysis.
Thus we can apply thread-local fixpoint analysis
while their unfolding suffers from path explosion due to local branching.
Furthermore, we can build unfoldings for general domains even with 
invalid independence relations while their approach is restricted to the 
independence encoded in the syntax of a Petri net and bitvector domains.


Compared to dynamic analysis of concurrent programs~\cite{AAJS14,FHRV13,KSH14,GLSW17},
our approach builds on top of a (super-)optimal partial-order reduction~\cite{RSSK15}
and is able to overcome a high degree of path explosion unrelated to thread
interference.



\section{Conclusion}
\label{sec:concl}

We introduced a new algorithm for static analysis of concurrent 
programs based on the combination of abstract interpretation 
and unfoldings.
Our algorithm explores an abstract unfolding using a new notion of
independence to avoid redundant transformer application in an optimal POR
strategy, thread-local fixed points to reduce the size of the unfolding, and
a novel cutoff criterion based on subsumption to guarantee termination of
the analysis.

Our experiments show that~\apoet~generates about 10x fewer
false positives than a mature thread modular abstract interpreter 
and is able to terminate on a large set of benchmarks as opposed to
solver-based tools that have the same precision.
We observed that the major reasons for the success of~\apoet~are: (1)~the
use of cutoffs to cope with and prune cyclic explorations caused by
spinlocks and (2)~\ccollapse~mitigates path explosion in the threads.
Our analyser is able to scale with the number of threads as long
as the interference between threads does not increase.
As future work, we plan to experimentally evaluate the application
of local widenings to force cutoffs to increase the scalability of 
our approach.

\subsubsection*{Acknowledgments.} 
The authors would like to thank Antoine Miné for the invaluable 
help with~\astreea and the anonymous reviewers for their helpful 
feedback.

\bibliographystyle{plain}
\bibliography{absint,concurrency,popl17}

\newpage
\appendix

\section{Proofs: Abstract Partial-Order Semantics}

\subsection{Proofs for \cref{sec:ind}: Independence of Transformers}

\liftedindep*

\begin{proof}
Let $f_1, f_2$ be transformers of~$\ccdom M$ such that $f_1 \indepp f_2$.
Let $a_1 \eqdef m_0^{-1} (f_1)$
and $a_2 \eqdef m_0^{-1} (f_2)$
be the corresponding program statements.
We know that $a_1 \indep a_2$.
Let $d \eqdef \set{s_1, \ldots, s_n} \in \reach{\ccdom M}$ be an element
of~$\ccdom M$. By definition of~$\ccdom M$ we know that~$d$ contains only
reachable states of~$M$.
Furthremore, we know that~$a_1$ and~$a_2$ commute on all of them.
Let $d_1 \eqdef f_1 \circ f_2 (d)$
and $d_2 \eqdef f_2 \circ f_1 (d)$
be the abstract elements obtained after executing the abstract
transformers in both orders. We need to show that $d_1 = d_2$.
\Wlog, we show that $d_1 \subseteq d_2$ (the opposite direction holds by
symmetry).
Let $s' \in d_1$ be an state in~$d_1$.
Then there is some $s \in d$ such that $s \fire{a_2,a_1} s'$.
Since $a_1$ and $a_2$ are independent under~$\indep$, then also
$s \fire{a_1,a_2} s'$ (by commutativity $a_1$ is enabled at $s$ because it was
at the state reached after executing~$a_2$ and both orderings reach~$s'$).
By definition of $f_1, f_2$ we get that $s' \in d_2$.
\end{proof}

\subsection{Results in Section 5: Unfolding Semantics}

This section contains the proofs of the formal statements made in
\cref{sec:unf}. All notations fixed in~\cref{sec:unf} are assumed here.
We will need to make some new definitions.

We recall now the definition of a PES.
An \emph{$X$-labelled prime event structure}~\cite{NPW81}
($X$-LPES, or PES in short)
is a tuple $\les \eqdef \tup{E, <, \cfl, h}$ where
${<} \subseteq E \times E$ is a strict partial order,
${\cfl} \subseteq E \times E$ is a symmetric, irreflexive relation, and 
$h \colon E \to X$ is a labelling function satisfying:
\begin{itemize}
\item
  for all $e \in E$, $\set{e' \in E \colon e' < e}$ is finite, and
  \eqtag{e:les1}
\item
  for all $e,e',e'' \in E$, if $e \cfl e'$ and $e' < e''$, then $e \cfl e''$.
  \eqtag{e:les2}
\end{itemize}

Event structures are naturally (partially) ordered by a \emph{prefix}
relation~$\ispref$.
Given two PESs
$\les \eqdef \tup{E, <, \cfl, h}$ and
$\les' \eqdef \tup{E', <', \cfl', h'}$,
we say that
$\les$ is a \emph{prefix} of~$\les'$, written $\les \ispref \les'$,
when
$E \subseteq E'$,
$<$ and $\cfl$ are the projections of $<'$ and $\cfl'$ to~$E$,
and $E \supseteq \set{e' \in E' \colon e' < e \land e \in E}$.
Moreover, the set of prefixes of a given PES~$\les$ equipped with~$\ispref$ is a
complete lattice.

\Cref{def:unf} defines $\unf{\ddom,\indep}$ using an iterative procedure that
constructs, possibly, an infinite object.
We call \emph{unfolding prefix} the structure~$\les$ that the algorithm
constructs after countably many steps, possibly before reaching fixpoint.

We now prove that the set of unfolding prefixes equipped with relation~$\ispref$
forms a complete \mbox{join-semilattice},
where the operator $\union \cdot$, defined below, is the join operator.
In turn, this implies the existence of a unique
\mbox{$\ispref$-maximal} element,
that will be found by \cref{def:unf} when it reaches fixpoint.

We first define the operator~$\union \cdot$.
Let
\[
P \eqdef \set{
\tup{E_1, <_1, \cfl_1, h_1},
\tup{E_2, <_2, \cfl_2, h_2},
\ldots}
\]
be a countable set of finite unfolding prefixes of~$\ddom$.
The \emph{union} of all of them is the PES
$\union P \eqdef \tup{E, {<}, {\cfl}, h}$, where
\[
E \eqdef \bigcup_{1 \le i} E_i
\qquad
{<} \eqdef \bigcup_{1 \le i} {<_i}
\qquad
h \eqdef \bigcup_{1 \le i} h_i,
\]
and $\cfl$ is the $\subseteq$-minimal relation on $E \times E$ that
satisfies \eref{les2} and such that
$e \cfl e'$ holds
for every two events $e, e' \in E$ iff
\begin{equation}
\label{e:union}
e \notin [e'] \text{ and }
e' \notin [e] \text{ and }
\lnot (h(e) \bowtie h(e')).
\end{equation}

Since every element of $P$ is a PES, clearly $\union P$ is also a PES.
\cref{e:les1} and \eref{les2} are trivially satisfied.
Notice that all events in $E_1, E_2, E_3, \ldots$ are pairs of the form
$\tup{t,H}$, and the union of two or more $E_i$'s will ``\textit{merge}'' many
equal events.

\begin{lemma}
\label{lem:prefixes.lub}
For every set $P$ of unfolding prefixes, $\union P$ is the least-upper bound
of~$P$ with respect to the order $\ispref$.
\end{lemma}

\begin{proof}
Let $F \eqdef \set{\ppref_i}$ with $i \in \nat$ be a countable set
of finite or infinite prefixes, where
$\ppref_i \eqdef \tup{E_i, <_i, \cfl_i, h_i}$.
In fact, we can assume that all $\ppref_i$ are finite.
Either $\ppref_i$ is finite or it has been constructed by the third rule
in~\cref{def:unf}.
In the second case, it is the $\union \cdot$ of a countably
infinite set of finite prefixes. In both cases, $\ppref_i$ accounts for
countably many finite prefixes.

Now, for any countable set of prefixes $X \eqdef \set{\qpref_i}$, $i \in \nat$,
it is immediate to show that
\[
\union{\ppref,\union X} = \union{\ppref, \qpref_1, \qpref_2, \ldots}.
\]
Finally, since the union of countably many countable sets is a countable set, we
can assume \wlogg that~$F$ is a countable set of finite prefixes
(this assumes the Axiom of choice).

Let $\ppref \eqdef \union F$ be their union, where
$\ppref \eqdef \tup{E, <, \cfl, h}$.
We need to show that
\begin{itemize}
\item (upper bound) $\ppref_i \ispref \ppref$;
\item (least element) for any unfolding prefix $\ppref'$ such that $\ppref_j \ispref \ppref'$
holds for all $1 \le j$,
we have that $\ppref \ispref \ppref'$.
\end{itemize}

We start showing that $\ppref$ is an upper bound. Let $\ppref_i \in F$ be an
arbitrary unfolding prefix. We show that $\ppref_i \ispref \ppref$:
\begin{itemize}
\item
  Trivially $E_i \subseteq E$.

\item
  ${<_i} \subseteq {<} \cap (E_i \times E_i)$. Trivial.
\item
  ${<_i} \supseteq {<} \cap (E_i \times E_i)$.
  Let $e, e' \in E$ be two events of $\ppref$.
  Assume that $e < e'$ and that both $e$ and $e'$ are in $E_i$.
  Since $e' \in E$, there is some $j \in \nat$ such that $e <_j e'$,
  and both $e$ and $e'$ are in $E_j$.
  Assume that $e' \eqdef \tup{t,H}$.
  Since $\ppref_j$ is a finite prefix constructed by \defref{unf},
  then necessarily $e \in H$.
  As a result, \defref{unf} must have found that
  $e'$ was in $H$ when adding $e$ to the prefix that eventually became
  $\ppref_i$, and consequently $e' <_i e$.

\item
  ${\cfl_i} \subseteq {\cfl} \cap (E_i \times E_i)$. Trivial.
\item
  ${\cfl_i} \supseteq {\cfl} \cap (E_i \times E_i)$.
  Assume that $e \cfl e'$ and that $e, e' \in E_i$.
  We need to prove that $e \cfl_i e'$.
  Assume \wlogg that $e'$ was added to $\ppref_i$ by \defref{unf}
  after~$e$.
  If $e$ and $e'$ satisfy \cref{e:union}, then trivially
  $e \cfl_i e'$.
  If not, then assume \wlogg that
  there exists some $e'' < e'$ such that
  $e \cfl e''$, and such that $e$ and $e''$ satisfy \cref{e:union}.
  Then $e \cfl_i e''$ and,
  since $\ppref_i$ is a LES 
  then we have $e \cfl_i e'$.
\item
  $h_i = h \cap (E_i \times E_i)$. Trivial.
\end{itemize}

We now focus on proving that $\ppref$ is the least element among the upper
bounds of $F$. Let $\ppref' \eqdef \tup{E', <', \cfl', h'}$ be an upper bound of all elements of $F$.
We show that $\ppref \ispref \ppref'$.
\begin{itemize}
\item
  Since $E$ is the union of all $E_i$ and all $E_i$ are by hypothesis
  in~$E'$, then necessarily $E \subseteq E'$.

\item
  ${<} \subseteq {<'} \cap (E \times E)$.
  Assume that $e < e'$. By definition $e$ and $e'$ are in $E$, so we only
  need to show that $e <' e'$.
  We know that there is some $i \in \nat$ such that $e <_i e'$.
  We also know that $\ppref_i \ispref \ppref'$, which implies that $e <' e'$.
\item
  ${<} \supseteq {<'} \cap (E \times E)$.
  Assume that $e < e'$ and that $e,e' \in E$.
  We know that there is some $i \in \nat$ such that $e,e' \in E_i$.
  We also know that $\ppref_i \ispref \ppref'$, which implies that
  ${<_i} = {<'} \cap (E_i \times E_i)$.
  This means that $e <_i e'$, and so $e < e'$.

\item
  $h = h' \cap (E \times E)$. Trivial.

\item
  ${\cfl} \subseteq {\cfl'} \cap (E \times E)$.
  Assume that $e \cfl e'$.
  Then $e$ and $e'$ are in $E$.
  Two things are possible. Either $e,e'$ satisfy \cref{e:union} or, \wlogg,
  there exists some $e'' < e'$ such that $e$ and $e''$ satisfy
  \cref{e:union}.
  In the former case, using items above, it is trivial to show
  that $\lnot (e <' e')$,
  that $\lnot (e' <' e)$, and
  that $h'(e) \depen h'(e')$.
  This means that $e \cfl' e'$.
  In the latter case its the same.

\item
  ${\cfl} \supseteq {\cfl'} \cap (E \times E)$. Trivial.
\end{itemize}
\end{proof}

\unfunique*

\begin{proof}
It is trivial to show that $\unf{\ddom,\indep}$ satisfies
\cref{e:les1} and
\cref{e:les2}.
When the procedure in \cref{def:unf} reaches fixpoint it obviously computes the
least-upper bound of the set of unfolding prefixes. That is a unique element in
the lattice.
\end{proof}

\unfsound*

\begin{proof}
Item (2) has already been proved in~\cite[Lemma 16]{RSSK15long}.

To prove item~(3) we assume that~(1) holds.
Item~(3) then holds as a consequence of the way in which the set
$\mathcal{H}_{\les,\indepx,f}$ of histories for an event is defined.
Assume that~$C \eqdef [e]$ is the local configuration of event~$e \eqdef \tup{f,H}$.
Since there is only one maximal event in $C$ (event~$e$), then
necessarily~$\sigma$ has the form $\sigma \eqdef \tilde\sigma . h(e)$,
for $\tilde\sigma \in \inter H$. From item~(1) we know that all interleavings
of~$H$ reach the same dataflow fact, and since $\state H$ is defined as the meet
of all of them, then necessarily $\state H = \state{\tilde\sigma}$.
We also know, from~\defref{unf} that
transformer~$f$ is enabled at~$\state H$.
This means that~$\state \sigma \ne \bot$.

Finally we prove~(1).
The proof is by induction on the size $|C|$ of the configuration.

\emph{Base case}. $|C| = 0$ and so, $C = \emptyset$. The set of interleavings
of~$C$ contains zero linearizations and the result trivially holds.

\emph{Inductive Step}.
Assume that the result holds for configuration of size $k - 1$
and assume that $|C| = k$.
Let $e \in C$ be any $<$-maximal event in $C$, and assume that
$\sigma$ and $\sigma'$ have the form
\begin{align*}
\sigma  &\eqdef \tilde\sigma, h(e), f_1, \ldots, f_l \\
\sigma' &\eqdef \tilde\sigma', h(e), g_1, \ldots, g_m.
\end{align*}
Recall that the interleavings of a configuration are the topological orderings
of events \wrt causality. As a result transformer $h(e)$ is independent
in~$\indep$ to the transformers that label all events
$f_1, \ldots, f_l$ and $g_1, \ldots, g_m$.

Now consider the dataflow fact $d \eqdef \state{\tilde\sigma}$. If $d = \bot$,
then clearly $\tilde\sigma, h(e), f_1$ reaches $\bot$ as well.
If $d \ne \bot$, then $\tilde\sigma$ is a run of~$\ddom$ and $d \in \reach
\ddom$. Now, by construction of~$\indep$, we have that $\tilde\sigma, f_1, h(e)$ is
also a run of~$\ddom$ and
\[
\state{\tilde\sigma, h(e), f_1} = \state{\tilde\sigma, f_1, h(e)}.
\]
Applying the same argument $l-1$ times more we prove that
$\tilde\sigma, f_1, \ldots, f_l, h(e)$ is a run of~$\ddom$ and that
\[
\state{\sigma} = \state{\tilde\sigma, f_1, \ldots, f_l, h(e)}.
\]
That is, we have ``\textit{pushed back}'' the occurrence of transition $h(e)$ in
the interleaving without changing the state ($\bot$ or not) reached by the
interleaving.
Using the same argument, this time applied to~$\sigma'$ instad of~$\sigma$, we
can also show that
\[
\state{\sigma'} = \state{\tilde\sigma', g_1, \ldots, g_m, h(e)}.
\]
Now, we remark that both
$\tilde\sigma f_1, \ldots, f_l$ and
$\tilde\sigma' g_1, \ldots, g_m$
are interleavings of $C \setminus \set e$, a configuration of size $k-1$.
By induction hypothesis both interleavings thus satisfy that
\[
\state{\tilde\sigma, f_1, \ldots, f_l} =
\state{\tilde\sigma', g_1, \ldots, g_m}
\]
It then follows that
\[
\begin{split}
\state{\sigma} &= \state{\tilde\sigma, f_1, \ldots, f_l, h(e)} \\
 &= \state{\tilde\sigma', g_1, \ldots, g_m, h(e)} \\
 &= \state{\sigma'}
\end{split}
\]
\end{proof}

\unfcomplete*

\begin{proof}
Assume that $\sigma$ fires at least one transition.
The proof is by induction on the length $|\sigma|$ of the run.

\emph{Base Case}.
If $\sigma$ fires one transformer~$f$, then~$f$ is enabled at $d_0$,
the initial dataflow fact of~$\ddom$.
Then $\set \dagger$ is a history for~$f$,
as necessarily $\state{\set \dagger}$ enables~$f$.
This means that $e \eqdef \tup{f, \set \dagger}$ is an event of~$\unf{\ddom,\indep}$, and
clearly $\sigma \in \inter{\set{\dagger, e}}$.
It is easy to see that no other event $e'$ different than $e$ but such that
$h(e) = h(e')$ can exist in $\unf{M,\indep}$ and satisfy that the history $\causes{e'}$
of~$e'$ equals the singleton $\set \dagger$.
The representative configuration for $\sigma$ therefore exists and is unique.

\emph{Inductive Step}.
Assume that $\sigma \eqdef \sigma' f$.
By the induction hypothesis, we assume that there exist a unique configuration
$C'$ such that $\sigma' \in \inter{C'}$.
By \thmref{unf.sound}, all sequences in $\inter{C'}$ reach the same dataflow
fact~$d$.
Furthermore, $\sigma'$ is one of them, and it is also a run in~$\runs \ddom$.
This implies that $d \ne \bot$.
It also implies that~$f$ is enabled at~$d$.

If all $<$-maximal events $e' \in C'$
satisfy that $h(e') \depen f$, then $C'$ is a history for transformer~$f$, and
$\unf{\ddom,\indep}$ contains an event $e \eqdef \tup{f,C'}$.
Let $C \eqdef C' \cup \set e = [e]$ be the configuration that contains~$C'$
and~$e$.
Clearly $\sigma \in \inter C$.
Below we show that such~$C$ is unique.

Alternatively, $C'$ could have one or more maximal events~$e'$ such that
$h(e') \indep f$. We now find a history for~$f$ inside of~$C'$, as follows.
Let $C'' \eqdef C' \setminus \set{e'}$, for any such $e'$,
and let $d'' \eqdef \state{C''}$.
Since~$f$ is enabled at $d$ and $f \indep h(e')$, then necessary~$f$ is also
enabled at~$d''$, as otherwise $f \circ h(e) (d'') = h(e) \circ f (d'')$ would
be~$\bot$, and~$f$ would not be enabled at~$d$.
If all $<$-maximal events of~$C''$ are dependent with~$f$, then $C''$ is a
history for~$f$, and we set $e \eqdef \tup{f,C''}$. If not, we can apply again
the argument a finite number of times (as $C'$ is finite) until we find a
history~$H$ for~$f$ inside of~$C'$ ($\set \dagger$ is always a valid history).
We set $e \eqdef \tup{f,H}$ and $C \eqdef C' \cup \set e$.
As before, clearly $\sigma \in \inter C$.

In both cases we found a configuration~$C \supseteq C'$ such that
$\sigma \in \inter C$.
We now argue that such~$C$ is unique.
By induction hypothesis we know that~$C'$ is the only configuration that
represents~$\sigma'$. 
If there was another $C''$ that represents~$\sigma$ and such that $C'
\not\subseteq C''$, then removing the maximal event that represents~$f$ in
$\sigma$ would yield a second representative for $\sigma'$.
This implies that any such $C''$ must include $C'$.
Showing uniqueness now reduces to showing
that~$e$ is the only event in
$\unf{\ddom,\indep}$ such that $C' \cup \set e$ represents~$\sigma$.

By contradiction, assume that $\unf{\ddom,\indep}$ contains another
event~$e' \eqdef \tup{f, H'}$
such that $C' \cup \set{e'}$ represents $\sigma$.
Assume that $e$ has the form $\tup{f,H}$.
Since $e \ne e'$ we know that $H \ne H'$.
By construction we know that $H \subseteq C'$.
If $H' \not\subseteq C'$, then $C' \cup \set{e'}$ would not be a configuration
(not causally closed).
So also $H' \subseteq C'$.
Now, since $H \ne H'$, \wlogg at least one of the maximal events in $H$ is not
in $H'$. Furthermore, that event is in~$C'$.
By \defref{unf} this means that $e'$ is in conflict with that event, and so
$C' \cup \set{e'}$ is not a configuration.
This is a contradiction.
\end{proof}

\subsection{Results in \cref{sec:abs-unf}: Abstract Unfoldings}

\absunfcorrect*

\begin{proof}
For the same reasons as in \thmref{unf.complete}, the statement of the theorem
restricts $\sigma$ to have at least one transformer.
The proof is by induction on the length $|\sigma|$ of the run.

\emph{Base Case}.
Run $\sigma$ fires one transformer~$f$ which is enabled at~$d_0$.
Let $\bar f \eqdef m(f)$ be the associated abstract transformer.
Then~$\bar f$ is enabled at~$\bar d_0$, and
$\set \dagger$ is a history for~$\bar f$.
As a result $\bar e \eqdef \tup{\bar f, \set \dagger}$ is an event
of~$\unf{\bar\ddom,\barindep}$, and
clearly $m(\sigma) = \bar f \in \inter{\set{\dagger, \bar e}}$.
It is immediate to show that $\set{\dagger, \bar e}$ is the only configuration that
represents $m(\sigma)$.

\emph{Inductive Step}.
Assume that $\sigma \eqdef \sigma' f$.
By the induction hypothesis, we assume that there exist a unique
configuration~$\bar C'$ in $\unf{\bar\ddom,\barindep}$ such that
$m(\sigma') \in \inter{\bar C'}$.

We fix some notation.
Let $\bar d' \eqdef \state{\bar C'}$ be the abstract state reached by $\bar C'$.
Let $d' \eqdef \state{\sigma'}$ be the concrete state reached by~$\sigma'$.
Let $\bar f \eqdef m(f)$ be the abstract counterpart of~$f$.

Now we show that $d' \sqsubseteq \gamma(\bar d')$.
Recall that $\bar d'$ is defined as the meet of the state reached by all
interleavings of~$\bar C'$.
Therefore, $\bar d'$ does not satisfy that $\state{m(\sigma')} \barsqsub \bar d'$,
which would probably be the easiest strategy to prove our goal.
We follow a different reasoning.
Since $m(\sigma') \in \inter{\bar C'}$, by
we get that
$m^{-1} (\inter{\bar C'})$ is a set of runs of the concrete domain~$\ddom$.
Furthermore, all those runs reach the same concrete dataflow fact~$d'$
as~$\sigma'$.
Then all runs in $\inter{\bar C'}$ reach abstract dataflow facts that soundly
approximate~$d'$.
What is more, in a Galois connection, the concretization map~$\gamma$ preserves abstract
meets. Formally, for any two abstract facts $\bar d_1, \bar d_2$, we have
\[
\gamma(\bar d_1) \meet \gamma(\bar d_2) ~\sle~
\gamma(\bar d_1 \barmeet \bar d_2).
\]
We thus can make the following development:
\[
\begin{split}
\gamma(\bar d')
  &= \gamma(\barbigmeet_{\bar\sigma \in \inter{\bar C'}} \state{\bar \sigma}) \\
  &\sge \bigmeet_{\bar\sigma \in \inter{\bar C'}} \gamma(\state{\bar\sigma}) \\
  &\sge \bigmeet_{\bar\sigma \in \inter{\bar C'}} d' \\
  &= d'
\end{split}
\]
This shows that $d' \sle \gamma(\bar d')$.
It also shows that $\bar f$ is enabled at $\bar d' = \state{\bar C'}$, since~$\bar f$
is a sound approximation of~$f$.

If all maximal events~$e$ of~$\bar C'$ are such that
$h(e) \bardepen \bar f$, then~$\bar C'$ is a history for~$\bar f$, and
$\bar e \eqdef \tup{\bar f, \bar C'}$ is an event of $\unf{\bar\ddom,\barindep}$.
Let $\bar C \eqdef \bar C' \cup \set{\bar e} = [\bar e]$
be the configuration that contains $\bar C'$ and $\bar e$.
Clearly $m(\sigma) \in \inter{\bar C}$.
Below we show that such~$\bar C$ is unique.

Alternatively, $\bar C'$ could have one or more maximal events independent
with~$\bar f$. Let~$\bar e'$ be any $<$-maximal event in~$\bar C'$ such that
$h(e') \barindep f$.
In the sequel we find a history for~$\bar f$ inside of~$C'$.
Let $\bar C'' \eqdef \bar C' \setminus \set{\bar e'}$,
and let $\bar d'' \eqdef \state{\bar C''}$.

We show that $\bar f$ is enabled at $\bar d''$.
Since all interleavings of~$\bar C'$ correspond to runs of~$\ddom$, necessarily
all interleavings of~$\bar C''$ are also executions of~$\ddom$;
and all of them reach the same dataflow fact, say~$d''$.
Using the same reasoning as above,
we can show that $d'' \sle \gamma(\bar d'')$.
Now, showing that $\bar f$ is enabled at $\bar d''$ reduces to showing that~$f$
is enabled at $d''$.
This, in turn, is a consequence of the fact that
$m^{-1}(h(\bar e')) (d'') = d'$ and $f(d') \ne \bot$ and the fact that
$m^{-1}(h(\bar e'))$ and $f$ are independent (we skip details).

This shows that $\bar f$ is enabled at $\state{\bar C''}$.
If all maximal events of $\bar C''$ are dependent with~$\bar f$, then~$\bar C''$
is a history for~$\bar f$.
If not, we can apply again
the argument a finite number of times (as $\bar C''$ is finite) until we find a
history~$H$ for~$\bar f$ inside of~$\bar C''$ ($\set \dagger$ is always a valid history).
We set $\bar e \eqdef \tup{\bar f,H}$ and
$\bar C \eqdef \bar C' \cup \set{\bar e}$.

In both cases we found a configuration~$\bar C \supseteq \bar C'$ such that
$m(\sigma) \in \inter{\bar C}$.
%

\cro{FIXME (temporarily fixed) continue here, should be a similar argument to
that of Theorem 2}

Showing that~$\bar C$ is unique requires the same reasoning than in
\thmref{unf.complete}, which we skip here.
\end{proof}

\remove{
\absunfruns*
\begin{proof}
Clearly \eref{absunf.runs2} is a consequence of \eref{absunf.runs1} and the
fact that every concrete execution $\sigma \in \runs \ddom$ has a representative
counterpart $m(\sigma) \in \runs{\bar\ddom}$ in the abstract domain.

Assume that $m^{-1} (\inter C)$ contains one execution
$\sigma$ from the set~$\runs \ddom$.
Recall that $\bar\indep$ is the pointwise lifting of
a weak relation~$\indep$ on~$\ddom$.
As a result, there exists a configuration $C'$ in~$\unf{\ddom,\indep}$
the unfolding of the concrete domain such that
$\sigma \in \inter{C'}$.
In fact,~$C$ and~$C'$ are isomorphic labelled partial orders, so the set of
interleavings of~$C$ is in bijective correspondance with the interleavings
of~$C'$.
Since $\sigma$ is a run of the concrete domain~$\ddom$,
by \thmref{unf.sound} all interleavings of~$C'$ are also runs of~$\ddom$.
This proves \eref{absunf.runs1}.
\end{proof}
}

%
%
%

\subsection{Results in \cref{sec:collapse}: Plugging Thread-Local Analysis}
\label{sec:app.collapse}

Our goal in this section is proving \cref{thm:collapse.correct}.
We will introduce some new notions necessary to formalize the operation
performed by the thread-local analysis.
In short, given~$\ddom$ and an implementation of \ccollapse, we will define a
new analysis instance
$\hat\ddom$, called the \emph{collapsing domain}, that we use to prove the
theorem.

For any global transformer~$f \in F^\text{glo}_i$,
we define the \emph{collapsing transformer}
$\hat f \colon D \to D$ of~$f$ as $\hat f \eqdef f \comp \collapse{i}$.
The set of collapsing transformers induce a new analysis
instance
\[
\hat \ddom \eqdef \tup{\hat D, \hatsle, \hat F, \hat d_0}
\]
that soundly approximates~$\ddom$, where~$\hat D$, $\hatsle$, and $\hat d_0$
are the same as in~$\ddom$, and the set of transformers is
$\hat F \eqdef \set{\hat f \in D \to D \colon f \text{ is global in } \ddom}$.
Our notion of approximation here is different than the one given in
\cref{sec:prelim}.
There we required exactly one abstract transformer per concrete one,
while~$\hat\ddom$ has only abstract transformers for the global concrete ones.
The notion of approximation here is rather suttering simulation of runs.
For any run $\sigma \in \runs \ddom$,
recall that $\hat\sigma$ is the subsequence of~$\sigma$ obtained by removing the local
transformers.
Clearly, $\hat\sigma \in \runs{\hat\ddom}$ when $\sigma \in \runs \ddom$.

Let $\indep$ be a weak independence on~$\ddom$.
One would wish that the collapsing transformers exhibits the same independence
than the original ones.
That is, if $f_1 \indep f_2$, then $\hat f_1$ and $\hat f_2$ commute on all
reachable facts.
Unfortunately, this is not true in general.
The abstract interpreter hidden behind $\collapse{\cdot,\cdot}$ might
apply widening on local loops, or data-flow joins could introduce imprecision
when merging paths,
all of which may affect the commutativity of the collapsing transformers.
In the sequel, we employ the pointwise lifting of relation~$\indep$
to~$\hat\ddom$, defined
as~${\hatindep} \eqdef
\set{\tup{\hat f_1, \hat f_2} \colon
f_1, f_2 \text{ are global and } f_1 \indep f_2}$.
As we said, relation~$\hatindep$ is in general not a weak independence relation
in~$\hat\ddom$.
However, using ideas similar to those of \cref{sec:abs-unf} we can proof the
following:

\begin{lemma}[Soundness of the abstraction]
\label{lem:dhat.correct}
For any execution $\sigma \in \runs \ddom$
there is a unique configuration~$C$ of
$\unf{\hat\ddom,\hatindep}$ such that $\hat\sigma \in \inter C$.
\end{lemma}

\begin{proof}

(\emph{Sketch})
The proof of this result is very similar to that of \thmref{absunf.correct}.

\emph{Base Case}.
Run $\sigma$ fires only local transformers and the length of $\hat\sigma$ is
zero. As in \thmref{absunf.correct} there is a unique representative
configuration
in~$\unf{\hat\ddom,\hatindep}$.

\emph{Inductive Step}.
Assume that $\sigma \eqdef \sigma' f$.
By the induction hypothesis, we assume that there exist a unique
configuration~$\bar C'$ in $\unf{\hat\ddom,\hatindep}$ such that
$\hat\sigma' \in \inter{C'}$.

We distinguish two cases
\begin{itemize}
\item
  Transformer $f$ is local.
  Then $\hat\sigma = \hat\sigma'$ and $C'$ is a representative configuration
  for $\sigma$.
\item
  Transformer $f$ is global.
  Then assume that $\sigma$ is of the form
  $\sigma = \sigma_g\sigma_l f$, where $\sigma_g$ ends in a global transformer
  and $\sigma_l$ contains only local transformers.
  Observe that $C'$ is also a representative configuration of~$\sigma_g$.

  Clearly,
  $\state{C'} \sge \state{\sigma_g}$.
  Since $\collapse{\cdot,\cdot}$ always overapproximates the execution of any
  arbitrary sequence of local transformers, it must also overapproximate the
  execution of $\sigma_l$ from $\state{\sigma_g}$.
  This proves that $\hat f$ is enabled at $\state{C'}$.

  If all maximal events in $C'$ are dependent with $\hat f$ in~$\hatindep$,
  then $C'$ is a history for $\hat f$ and $e \eqdef \tup{\hat f, C'}$
  an event
  of~$\unf{\hat\ddom,\hatindep}$.
  If not, using the same reasoning as in \thmref{absunf.correct}
  we can find a history $H \subseteq C'$ for $\hat f$, and define event
  $e \eqdef \tup{\hat f, H}$.

  In both cases, by construction $C \eqdef C' \cup \set e$ is a configuration
  of~$\unf{\hat\ddom,\hatindep}$.
  Configuration~$C$ is a representative of $\hat \sigma$.
\end{itemize}
\end{proof}

We can now easily prove the main theorem of the section.

\collapsecorrect*

\begin{proof}
A call to \unfold{$\ddom,\indep,n$} with cutoff checking disabled
computes the unfolding of $\hat\ddom$, so we have that
\[
\pref{\ddom,\indep} =
\unf{\hat\ddom,\hatindep}.
\]
The theorem holds as a consequence of \cref{lem:dhat.correct}.
\end{proof}

\subsection{Formalizing Cutoff Events}
\label{sec:app.cutoffs}

In this section, a new cutoff criterion is defined that exploits the lattice
order~$\sle$ for more aggressive pruning than standard cutoffs.
Let~$\ddom$ be an analysis instance and~$\indep$ a weak independence.

In order to prune the unfolding, we need to refer to the order in which it is
constructed.
A \emph{strategy} is any strict (partial) order~$\prec$ on the finite
configurations of~$\unf{\ddom,\indep}$ satisfying that
when $C \subseteq C'$, then $C \prec C'$.
In other words, strategies refine the natural order in which the domain is
unfolded.

Each strategy identifies a set of feasible and cutoff events.
Intuitively, feasible events will be those which have no cutoff among the set
of causal predecessors:

\begin{definition}[Cutoffs]
\label{def:cutoffs}
An event $e$ of $\unf{\ddom,\indep}$ is \emph{$\prec$-feasible}
if all causal predecessors $e' \in \causes e$ are not $\prec$-cutoff.
A $\prec$-feasible event is a $\prec$-cutoff if
there exists some $\prec$-feasible event $e'$ in $\unf{\ddom,\indep}$,
called the \emph{corresponding event}, such that
$[e'] \prec [e]$ and
\begin{equation}
  \label{e:cutoff}
  \state{[e]} \sle \state{[e']}.
\end{equation}
\end{definition}

In other words, $e$ will be a cutoff iff the fact reached by the branch it
represents (local configuration) has already been ``\textit{seen}'' when~$\ddom$
is unfolded in the order stated by~$\prec$.
Observe that ``\textit{seen}'' formally means that another equally or less
precise element has been unfolded before.

While the notion of cutoffs has been around for a while in the literature of
unfoldings~\cite{Mcm93,ERV02,BHKTV14},
to the best of our knowledge, \defref{cutoffs} is the first to use a 
subsumption relation to match the corresponding event.
The most general previous definition~\cite{BHKTV14} only allowed states to be
compared using equivalence relations in~\eref{cutoff},
while we used the partial order~$\sle$.
The set of $\prec$-feasible events defines an unfolding prefix of~$\uunf$:

\begin{definition}[Feasible prefix]
\label{def:feasible.prefix}
The $\prec$-prefix of~$\ddom$ is the unique unfolding prefix
$\precpref{\ddom,\indep}$
of~$\unf{\ddom,\indep}$ that contains exactly all $\prec$-feasible events which
are not $\prec$-cutoffs.
\end{definition}

The shape and properties of the $\prec$-prefix strongly depend on the underlying
strategy~$\prec$.
One is interested in strategies that identify complete prefixes.

A well-known unfolding strategy is the
\emph{size order}~${\prec_s} \subseteq E \times E$, defined by
Ken McMillan in his seminal paper~\cite{Mcm93} as
$C \prec_s C'$ iff $|C| < |C'|$.
Adequate strategies~\cite{ERV02,BHKTV14} were discovered later and yield up to
exponentially smaller prefixes.
In order to keep the presentation concise, we restrict next theorem to
the size order (although it also holds for adequate strategies).

\begin{restatable}{theorem}{prefixcomplete}
\label{thm:complete}
The unfolding prefix $\pref{\ddom,\indep}^{\prec_s}$ is $\ddom$-complete.
\end{restatable}

\begin{proof}
(\emph{Sketch})
Let $d \in \reach \ddom$ be a reachable state.
Then there is some configuration $C$ in $\unf{\ddom,\indep}$ such that
$d = \state C$.
If $C$ is free of $\prec$-cutoff events, then $C$ is
in~$\pref{\ddom,\indep}^{\prec_s}$ and we found the configuration that we
searched.

If not, let $e \in C$ be a $\prec$-cutoff event and
$e'$ the corresponding event in $\unf{\ddom,\indep}$.
Since $d \ne \bot$, any interleaving of~$C$ is a run of~$\ddom$.
Since $\state{[e']} \sge \state{[e]}$,
any interleaving of $[e]$ can be extended with the transformers that
label in any topological sorting of the events in $C \setminus [e]$,
and the resulting sequence is a run $\sigma' \in \runs \ddom$
that satisfies $d \sle \state{\sigma'}$.
Furthermore, since all runs of~$\ddom$ are represented as (unique)
configurations of $\unf{\ddom,\indep}$, it is possible to extend
configuration~$[e]$ into a unique configuration that represents $\sigma'$.
Let it be~$C'$. We have that $d = \state C \sle \state{C'}$, and
$C'$ is at least one event smaller than~$C$.

If $C'$ has no cutoff, then we found the configuration that we were searching.
If not, we only need to repeat this argument a finite number of times (since
every time we remove at least one event from the configuration) until we find a
configuration that reaches a state that covers~$d$.
\end{proof}

\subsection{Pruning without Relaxed Independence}
\label{sec:limcutoffs}

In \secref{abs-unf} we unfolded an abstract domain~$\bar\ddom$ into an event
structure
$\unf{\bar\ddom,\barindep}$ under an independence relation~$\barindep$ that was
weak in the concrete domain~$\ddom$ but non-weak
in the abstract one~$\bar\ddom$.

It~would be natural to extend the cutoff criterion introduced above,
which requires a weak independence, to employ the non-weak
relation~$\barindep$.
Unfortunately, in this case the feasible $\prec$-prefix is not necessarily
complete.
The proof of \thmref{complete} relies on the fact that all 
runs of~$\bar\ddom$ will appear under the form of one configuration
in~$\unf{\bar\ddom,\barindep}$.
However, the non-weak relation~$\barindep$ fails to guarantee that.

Alternatively, one may try to change the completeness criterion,
asking that all facts reachable in the concrete domain~$\ddom$ are present in
the unfolding of the abstract domain.
Unfortunately, proving \thmref{complete} with this notion of completeness
fails again, for the same reason.
The reasoning behind the notion of cutoff events
fundamentally relies on the fact that arbitrary executions of~$\bar\ddom$ must
be present in~$\unf{\bar\ddom,\barindep}$.

Therefore, using cutoff criteria for the abstract unfolding is possible
only together with weak independence relations.
Fortunately, at least for simple domains such as intervals, computing a weak
independence seems to be reasonably inexpensive.

\subsection{Results in \cref{sec:cutoffs}: Cutoff Events}

In the following proof we make use of the collapsing domain introduced
in~\cref{sec:app.collapse}.

\algoprefixcomplete*

\begin{proof}
If $\indep$ respects independence, then it is straightforward to show
that~$\hatindep$ is a weak independence in $\hat\ddom$.
That means that \cref{thm:complete} is applicable and implies that
$\preff{\ddom,\indep}$ is $\ddom$-complete,
as \cref{a:a1} computes exactly that prefix.
\end{proof}

\end{document}